\documentclass[sn-mathphys-ay,pdflatex,%
 colorlinks,linkcolor=blue,citecolor=blu]{sn-jnl}%

\usepackage{graphicx}%y
\usepackage{multirow}%
\usepackage{amsmath,amssymb,amsfonts}%
\usepackage{amsthm}%
\usepackage{mathrsfs}%
\usepackage[title]{appendix}%
\usepackage{xcolor}%
\usepackage{textcomp}%
\usepackage{manyfoot}%
\usepackage{booktabs}%
\usepackage{algorithm}%
\usepackage{algorithmicx}%
\usepackage{algpseudocode}%
\usepackage{listings}%

\usepackage{agda}%
\usepackage[sans]{dsfont}%
\usepackage[mathscr]{eucal}%
\usepackage{newunicodechar}%
\usepackage{setspace}%
\usepackage{stmaryrd}%
\usepackage{url}%

\theoremstyle{thmstyleone}%
\newtheorem{theorem}{Theorem}[section]%
\newtheorem{lemma}[theorem]{Lemma}%

\theoremstyle{thmstyletwo}%

\theoremstyle{thmstylethree}%
\newtheorem{definition}[theorem]{Definition}%
\newtheorem{example}[theorem]{Example}%
\newtheorem{remark}[theorem]{Remark}%

\renewcommand{\AgdaFontStyle}[1]{\texttt{#1}}%
\renewcommand{\AgdaKeywordFontStyle}[1]{\texttt{#1}}%
\renewcommand{\AgdaBoundFontStyle}[1]{\emph{\texttt{#1}}}%
\AgdaNoSpaceAroundCode{}%

\newunicodechar{ℕ}{\ensuremath{\mathds{N}}}
\newunicodechar{𝔸}{\ensuremath{\mathds{A}}}
\newunicodechar{λ}{\ensuremath{\mathnormal{}\lambda}}
\newunicodechar{𝛌}{\ensuremath{\mathbf{\lambda}}}
\newunicodechar{Λ}{\ensuremath{\mathbf{\Lambda}}}
\newunicodechar{Γ}{\ensuremath{\mathnormal{}\Gamma}}
\newunicodechar{Σ}{\ensuremath{\mathnormal{}\Sigma}}
\newunicodechar{ν}{\ensuremath{\mathnormal{}\nu}}
\newunicodechar{∅}{\ensuremath{\emptyset}}
\newunicodechar{→}{\ensuremath{\mathnormal\shortrightarrow}}
\newunicodechar{≡}{\ensuremath{\mathnormal\equiv}}
\newunicodechar{∀}{\ensuremath{\mathnormal\forall}}
\newunicodechar{И}{\ensuremath{\mathnormal\New}}
\newunicodechar{∑}{\ensuremath{\mathnormal\textstyle\sum}}
\newunicodechar{×}{\ensuremath{\mathnormal{\times}}}
\newunicodechar{∪}{\ensuremath{\mathnormal\cup}}
\newunicodechar{′}{\ensuremath{\mathnormal{}'}}
\newunicodechar{₀}{\ensuremath{{}_0}}
\newunicodechar{₁}{\ensuremath{{}_1}}
\newunicodechar{₂}{\ensuremath{\mathnormal{}_2}}
\newunicodechar{₃}{\ensuremath{\mathnormal{}_3}}
\newunicodechar{₄}{\ensuremath{\mathnormal{}_4}}
\newunicodechar{₅}{\ensuremath{\mathnormal{}_5}}
\newunicodechar{₆}{\ensuremath{\mathnormal{}_6}}
\newunicodechar{₇}{\ensuremath{\mathnormal{}_7}}
\newunicodechar{₈}{\ensuremath{\mathnormal{}_8}}
\newunicodechar{₉}{\ensuremath{\mathnormal{}_9}}
\newunicodechar{₊}{\ensuremath{{}_{+}}}
\newunicodechar{≐}{\ensuremath{\mathnormal\doteq}}
\newunicodechar{≠}{\ensuremath{\mathnormal\not=}}
\newunicodechar{≢}{\ensuremath{\mathnormal\not\equiv}}
\newunicodechar{∉}{\ensuremath{\mathnormal\notin}}
\newunicodechar{∈}{\ensuremath{\mathnormal\in}}
\newunicodechar{∶}{\ensuremath{\mathnormal :}}
\newunicodechar{⦂}{\ensuremath{\mathnormal :}}
\newunicodechar{⨟}{\ensuremath{\mathnormal ;}}
\newunicodechar{∙}{\ensuremath{\mathnormal\cdot}}
\newunicodechar{≻}{\ensuremath{\mathnormal\succ}}
\newunicodechar{≤}{\ensuremath{\mathnormal\leq}}
\newunicodechar{≥}{\ensuremath{\mathnormal\geq}}
\newunicodechar{∼}{\ensuremath{\mathnormal\sim}}
\newunicodechar{⊢}{\ensuremath{\mathnormal\vdash}}
\newunicodechar{｛}{\ensuremath{\mathnormal\{}}
\newunicodechar{｝}{\ensuremath{\mathnormal\}}}
\newcommand{\lbbar}{\{\kern-0.5ex|}
\newunicodechar{⦃}{\ensuremath{\mathnormal\lbbar}}
\newcommand{\rbbar}{|\kern-0.5ex\}}
\newunicodechar{⦄}{\ensuremath{\mathnormal\rbbar}}
\newunicodechar{‿}{\ensuremath{\mathnormal\smallsmile}}
\newunicodechar{≔}{\ensuremath{\mathnormal{:}{=}}}
\newunicodechar{∣}{\ensuremath{\mathnormal |}}
\newunicodechar{‼}{\ensuremath{\mathnormal !\;!}}
\newunicodechar{∞}{\ensuremath{\mathnormal\infty}}
\newunicodechar{𝐚}{\ensuremath{\mathnormal\mathbf{a}}}
\newunicodechar{𝐜}{\ensuremath{\mathnormal\mathbf{c}}}
\newunicodechar{𝐞}{\ensuremath{\mathnormal\mathbf{e}}}
\newunicodechar{𝐢}{\ensuremath{\mathnormal\mathbf{i}}}
\newunicodechar{𝐧}{\ensuremath{\mathnormal\mathbf{n}}}
\newunicodechar{𝐍}{\ensuremath{\mathnormal\mathbf{N}}}
\newunicodechar{𝐨}{\ensuremath{\mathnormal\mathbf{o}}}
\newunicodechar{𝐫}{\ensuremath{\mathnormal\mathbf{r}}}
\newunicodechar{𝐬}{\ensuremath{\mathnormal\mathbf{s}}}
\newunicodechar{𝐭}{\ensuremath{\mathnormal\mathbf{t}}}
\newunicodechar{𝐮}{\ensuremath{\mathnormal\mathbf{u}}}
\newunicodechar{𝐯}{\ensuremath{\mathnormal\mathbf{v}}}
\newunicodechar{𝐳}{\ensuremath{\mathnormal\mathbf{z}}}
\newunicodechar{𝐎}{\ensuremath{\mathnormal\mathbf{0}}}
\newunicodechar{⟦}{\ensuremath{\mathnormal\llbracket}}
\newunicodechar{⟧}{\ensuremath{\mathnormal\rrbracket}}
\newunicodechar{∼}{\ensuremath{\mathnormal\sim}}
\newunicodechar{⁻}{\ensuremath{\mathnormal{}^{-}}}

\makeatletter
\providecommand{\leftsquigarrow}{%
  \mathrel{\mathpalette\reflect@squig\relax}%
}
\newcommand{\reflect@squig}[2]{%
  \reflectbox{$\m@th#1\rightsquigarrow$}%
}
\makeatother

\newcommand{\abs}[1]{\mathop{#1.}}
\newcommand{\actFin}{\mathtt{actFin}}
\newcommand{\aeq}{\sim}
\newcommand{\aeqa}{\sim^{\mathrm{a}}}
\newcommand{\aeqb}{\sim^{\mathrm{b}}}
\newcommand{\agdadef}[3]{\href{\agdaroot#1.html\##2}{\texttt{#3}}}
\newcommand{\agdalink}[1]{\href{\agdaroot#1.html}{\texttt{#1}}}
\newcommand{\agdalinkalt}[2]{\href{\agdaroot#1.html}{\texttt{#2}}}
\newcommand{\agdaroot}{https://amp12.github.io/WSLN/html/}
\newcommand{\arity}{\mathop{\mathtt{ar}}}
\newcommand{\Arg}[1]{\mathtt{Arg}[#1]}
\newcommand{\atm}{\mathbf{a}\,}
\newcommand{\Atom}{\mathds{A}}

\newcommand{\bndNil}{\langle\rangle}
\newcommand{\bndPair}[2]{\langle#1\mathbin{,}#2\rangle}
\newcommand{\close}{\mathbin{\leftsquigarrow}}
\newcommand{\centeredcirc}[1]{\vcenter{\hbox{$#1\circ$}}}
\newcommand{\smallcirc}{
  \mathbin{\mathchoice{
      \centeredcirc\scriptstyle}{
      \centeredcirc\scriptstyle}{
      \centeredcirc\scriptscriptstyle}{
      \centeredcirc\scriptscriptstyle}}}
\newcommand{\comp}{\smallcirc}
\newcommand{\conc}[1]{[#1]}

\newcommand{\CONS}{\mathbin{{:}{:}}}

\newcommand{\den}[1]{\llbracket#1\rrbracket}
\newcommand{\denb}[1]{\llbracket#1\rrbracket^{\mathrm{b}}}
\newcommand{\dens}[1]{\llbracket#1\rrbracket^{\mathrm{s}}}

\newcommand{\Empty}{\emptyset}
\newcommand{\ent}{\vdash}
\newcommand{\eq}{\equiv}
\newcommand{\Fin}{\mathop{\mathtt{Fin}}}
\newcommand{\freshfor}{\mathrel{\#}}
\newcommand{\FsetA}{\mathtt{Fset}\Atom}
\newcommand{\fun}{\shortrightarrow}

\newcommand{\idx}{\mathbf{i}\,}

\newcommand{\ins}[2]{\mathtt{insert}\,#1\,#2}

\newcommand{\List}{\mathop{\mathtt{List}}}

\newcommand{\Nat}{\mathds{N}}
\newcommand{\new}{\mathtt{new}}
\newcommand{\newNotIn}{\new{\notin}}
\newcommand{\New}{\reflectbox{$\mathsf{N}$}}
\newcommand{\NIL}{\mathtt{[]}}
\newcommand{\NomArg}{\mathtt{NomArg}}
\newcommand{\NomBnd}{\mathtt{NomBind}}
\newcommand{\NomTrm}{\mathtt{NomTrm}}
\newcommand{\Op}{\mathop{\mathtt{Op}}}
\newcommand{\op}{\mathit{op}}
\newcommand{\oprn}{\mathbf{o}\,}
\newcommand{\open}{\mathbin{\rightsquigarrow}}
\newcommand{\rem}[2]{\mathtt{remove}\,#1\,#2}
\newcommand{\subst}{\ast}
\newcommand{\Set}{\mathtt{Set}}
\newcommand{\SetOne}{\Set_1}
\newcommand{\size}{\mathtt{size}}
\newcommand{\ssubst}{\mathbin{{:}{=}}}
\newcommand{\suc}{\mathop{\mathtt{suc}}}
\newcommand{\suchat}[1]{\mathtt{suc}^{#1}\,}
\newcommand{\supp}{\mathop{\mathtt{supp}}}
\newcommand{\toN}{\mathtt{to}\Nat\,}
\newcommand{\Trm}[1]{\mathtt{Trm}[#1]}
\newcommand{\TRM}{\mathtt{Trm}}

\newcommand{\wkn}{\mathbin{\smallsmile}}
\newcommand{\zero}{\mathtt{zero}}

\newcommand{\FIN}{
\begin{code}%
\>[0]\AgdaKeyword{data}\AgdaSpace{}%
\AgdaDatatype{Fin}\AgdaSpace{}%
\AgdaSymbol{:}\AgdaSpace{}%
\AgdaDatatype{ℕ}\AgdaSpace{}%
\AgdaSymbol{→}\AgdaSpace{}%
\AgdaPrimitive{Set}\AgdaSpace{}%
\AgdaKeyword{where}\<%
\\
\>[0][@{}l@{\AgdaIndent{0}}]%
\>[2]\AgdaInductiveConstructor{zero}\AgdaSpace{}%
\AgdaSymbol{:}\AgdaSpace{}%
\AgdaSymbol{\{}\AgdaBound{n}\AgdaSpace{}%
\AgdaSymbol{:}\AgdaSpace{}%
\AgdaDatatype{ℕ}\AgdaSymbol{\}}\AgdaSpace{}%
\AgdaSymbol{→}\AgdaSpace{}%
\AgdaDatatype{Fin}\AgdaSpace{}%
\AgdaSymbol{(}\AgdaInductiveConstructor{1+}\AgdaSpace{}%
\AgdaBound{n}\AgdaSymbol{)}\<%
\\
\>[2]\AgdaInductiveConstructor{suc}\AgdaSpace{}%
\AgdaSymbol{:}\AgdaSpace{}%
\AgdaSymbol{\{}\AgdaBound{n}\AgdaSpace{}%
\AgdaSymbol{:}\AgdaSpace{}%
\AgdaDatatype{ℕ}\AgdaSymbol{\}(}\AgdaBound{i}\AgdaSpace{}%
\AgdaSymbol{:}\AgdaSpace{}%
\AgdaDatatype{Fin}\AgdaSpace{}%
\AgdaBound{n}\AgdaSymbol{)}\AgdaSpace{}%
\AgdaSymbol{→}\AgdaSpace{}%
\AgdaDatatype{Fin}\AgdaSpace{}%
\AgdaSymbol{(}\AgdaInductiveConstructor{1+}\AgdaSpace{}%
\AgdaBound{n}\AgdaSymbol{)}\<%
\\
\\[\AgdaEmptyExtraSkip]%
\>[0]\AgdaComment{--\ Fin\ n\ as\ an\ initial\ segment\ of\ ℕ}\<%
\\
\>[0]\AgdaFunction{toℕ}\AgdaSpace{}%
\AgdaSymbol{:}\AgdaSpace{}%
\AgdaSymbol{∀\{}\AgdaBound{n}\AgdaSymbol{\}}\AgdaSpace{}%
\AgdaSymbol{→}\AgdaSpace{}%
\AgdaDatatype{Fin}\AgdaSpace{}%
\AgdaBound{n}\AgdaSpace{}%
\AgdaSymbol{→}\AgdaSpace{}%
\AgdaDatatype{ℕ}\<%
\\
\>[0]\AgdaFunction{toℕ}\AgdaSpace{}%
\AgdaInductiveConstructor{zero}%
\>[12]\AgdaSymbol{=}\AgdaSpace{}%
\AgdaNumber{0}\<%
\\
\>[0]\AgdaFunction{toℕ}\AgdaSpace{}%
\AgdaSymbol{(}\AgdaInductiveConstructor{suc}\AgdaSpace{}%
\AgdaBound{i}\AgdaSymbol{)}\AgdaSpace{}%
\AgdaSymbol{=}\AgdaSpace{}%
\AgdaInductiveConstructor{1+}\AgdaSpace{}%
\AgdaSymbol{(}\AgdaFunction{toℕ}\AgdaSpace{}%
\AgdaBound{i}\AgdaSymbol{)}\<%
\\
\\[\AgdaEmptyExtraSkip]%
\>[0]\AgdaComment{--\ The\ less-than-or-equal\ relation,\ ≤}\<%
\\
\>[0]\AgdaKeyword{infix}\AgdaSpace{}%
\AgdaNumber{4}\AgdaSpace{}%
\AgdaOperator{\AgdaDatatype{\AgdaUnderscore{}≤\AgdaUnderscore{}}}\<%
\\
\>[0]\AgdaKeyword{data}\AgdaSpace{}%
\AgdaOperator{\AgdaDatatype{\AgdaUnderscore{}≤\AgdaUnderscore{}}}\AgdaSpace{}%
\AgdaSymbol{:}\AgdaSpace{}%
\AgdaDatatype{ℕ}\AgdaSpace{}%
\AgdaSymbol{→}\AgdaSpace{}%
\AgdaDatatype{ℕ}\AgdaSpace{}%
\AgdaSymbol{→}\AgdaSpace{}%
\AgdaPrimitive{Set}\AgdaSpace{}%
\AgdaKeyword{where}\<%
\\
\>[0][@{}l@{\AgdaIndent{0}}]%
\>[2]\AgdaKeyword{instance}\<%
\\
\>[2][@{}l@{\AgdaIndent{0}}]%
\>[4]\AgdaInductiveConstructor{0≤}%
\>[8]\AgdaSymbol{:}\AgdaSpace{}%
\AgdaSymbol{∀\{}\AgdaBound{n}\AgdaSymbol{\}}%
\>[17]\AgdaSymbol{→}\AgdaSpace{}%
\AgdaNumber{0}\AgdaSpace{}%
\AgdaOperator{\AgdaDatatype{≤}}\AgdaSpace{}%
\AgdaBound{n}\<%
\\
\>[4]\AgdaInductiveConstructor{1+≤}\AgdaSpace{}%
\AgdaSymbol{:}\AgdaSpace{}%
\AgdaSymbol{∀\{}\AgdaBound{m}\AgdaSpace{}%
\AgdaBound{n}\AgdaSymbol{\}}\AgdaSpace{}%
\AgdaSymbol{→}\AgdaSpace{}%
\AgdaSymbol{⦃}\AgdaSpace{}%
\AgdaBound{\AgdaUnderscore{}}\AgdaSpace{}%
\AgdaSymbol{:}\AgdaSpace{}%
\AgdaBound{m}\AgdaSpace{}%
\AgdaOperator{\AgdaDatatype{≤}}\AgdaSpace{}%
\AgdaBound{n}\AgdaSpace{}%
\AgdaSymbol{⦄}\AgdaSpace{}%
\AgdaSymbol{→}\AgdaSpace{}%
\AgdaInductiveConstructor{1+}\AgdaSpace{}%
\AgdaBound{m}\AgdaSpace{}%
\AgdaOperator{\AgdaDatatype{≤}}\AgdaSpace{}%
\AgdaInductiveConstructor{1+}\AgdaSpace{}%
\AgdaBound{n}\<%
\end{code}}

\newcommand{\ACTFIN}{
\begin{code}%
\>[0]\AgdaComment{--\ Fin\ m\ as\ an\ initial\ segment\ of\ Fin\ n\ when\ m\ ≤\ n}\<%
\\
\>[0]\AgdaFunction{actFin}\AgdaSpace{}%
\AgdaSymbol{:}\AgdaSpace{}%
\AgdaSymbol{∀\{}\AgdaBound{m}\AgdaSymbol{\}}\AgdaSpace{}%
\AgdaSymbol{→}\AgdaSpace{}%
\AgdaDatatype{Fin}\AgdaSpace{}%
\AgdaBound{m}\AgdaSpace{}%
\AgdaSymbol{→}\AgdaSpace{}%
\AgdaSymbol{∀}\AgdaSpace{}%
\AgdaBound{n}\AgdaSpace{}%
\AgdaSymbol{→}\AgdaSpace{}%
\AgdaSymbol{⦃}\AgdaSpace{}%
\AgdaBound{\AgdaUnderscore{}}\AgdaSpace{}%
\AgdaSymbol{:}\AgdaSpace{}%
\AgdaBound{m}\AgdaSpace{}%
\AgdaOperator{\AgdaDatatype{≤}}\AgdaSpace{}%
\AgdaBound{n}\AgdaSpace{}%
\AgdaSymbol{⦄}%
\>[45]\AgdaSymbol{→}\AgdaSpace{}%
\AgdaDatatype{Fin}\AgdaSpace{}%
\AgdaBound{n}\<%
\\
\>[0]\AgdaFunction{actFin}\AgdaSpace{}%
\AgdaInductiveConstructor{zero}%
\>[15]\AgdaSymbol{(}\AgdaInductiveConstructor{1+}\AgdaSpace{}%
\AgdaSymbol{\AgdaUnderscore{})}%
\>[30]\AgdaSymbol{=}\AgdaSpace{}%
\AgdaInductiveConstructor{zero}\<%
\\
\>[0]\AgdaFunction{actFin}\AgdaSpace{}%
\AgdaSymbol{(}\AgdaInductiveConstructor{suc}\AgdaSpace{}%
\AgdaBound{i}\AgdaSymbol{)}\AgdaSpace{}%
\AgdaSymbol{(}\AgdaInductiveConstructor{1+}\AgdaSpace{}%
\AgdaBound{n}\AgdaSymbol{)}\AgdaSpace{}%
\AgdaSymbol{⦃}\AgdaSpace{}%
\AgdaInductiveConstructor{1+≤}\AgdaSpace{}%
\AgdaSymbol{⦄}\AgdaSpace{}%
\AgdaSymbol{=}\AgdaSpace{}%
\AgdaInductiveConstructor{suc}\AgdaSpace{}%
\AgdaSymbol{(}\AgdaFunction{actFin}\AgdaSpace{}%
\AgdaBound{i}\AgdaSpace{}%
\AgdaBound{n}\AgdaSymbol{)}\<%
\end{code}}

\newcommand{\TRMACT}{
\begin{code}%
\>[2]\AgdaFunction{Trm‿}\AgdaSpace{}%
\AgdaSymbol{:}\AgdaSpace{}%
\AgdaSymbol{\{}\AgdaBound{m}\AgdaSpace{}%
\AgdaSymbol{:}\AgdaSpace{}%
\AgdaDatatype{ℕ}\AgdaSymbol{\}}\AgdaSpace{}%
\AgdaSymbol{→}\AgdaSpace{}%
\AgdaOperator{\AgdaDatatype{Trm[}}\AgdaSpace{}%
\AgdaBound{m}\AgdaSpace{}%
\AgdaOperator{\AgdaDatatype{]}}\AgdaSpace{}%
\AgdaSymbol{→}\AgdaSpace{}%
\AgdaSymbol{∀}\AgdaSpace{}%
\AgdaBound{n}\AgdaSpace{}%
\AgdaSymbol{→}\AgdaSpace{}%
\AgdaSymbol{⦃}\AgdaSpace{}%
\AgdaSymbol{\AgdaUnderscore{}}\AgdaSpace{}%
\AgdaSymbol{:}\AgdaSpace{}%
\AgdaBound{m}\AgdaSpace{}%
\AgdaOperator{\AgdaDatatype{≤}}\AgdaSpace{}%
\AgdaBound{n}\AgdaSpace{}%
\AgdaSymbol{⦄}\AgdaSpace{}%
\AgdaSymbol{→}\AgdaSpace{}%
\AgdaOperator{\AgdaDatatype{Trm[}}\AgdaSpace{}%
\AgdaBound{n}\AgdaSpace{}%
\AgdaOperator{\AgdaDatatype{]}}\<%
\\
\>[2]\AgdaFunction{Arg‿}\AgdaSpace{}%
\AgdaSymbol{:}\AgdaSpace{}%
\AgdaSymbol{\{}\AgdaBound{m}\AgdaSpace{}%
\AgdaSymbol{:}\AgdaSpace{}%
\AgdaDatatype{ℕ}\AgdaSymbol{\}\{}\AgdaBound{ns}\AgdaSpace{}%
\AgdaSymbol{:}\AgdaSpace{}%
\AgdaDatatype{List}\AgdaSpace{}%
\AgdaDatatype{ℕ}\AgdaSymbol{\}}\AgdaSpace{}%
\AgdaSymbol{→}\AgdaSpace{}%
\AgdaOperator{\AgdaDatatype{Arg[}}\AgdaSpace{}%
\AgdaBound{m}\AgdaSpace{}%
\AgdaOperator{\AgdaDatatype{]}}\AgdaSpace{}%
\AgdaBound{ns}\AgdaSpace{}%
\AgdaSymbol{→}\AgdaSpace{}%
\AgdaSymbol{∀}\AgdaSpace{}%
\AgdaBound{n}\AgdaSpace{}%
\AgdaSymbol{→}\AgdaSpace{}%
\AgdaSymbol{⦃}\AgdaSpace{}%
\AgdaSymbol{\AgdaUnderscore{}}\AgdaSpace{}%
\AgdaSymbol{:}\AgdaSpace{}%
\AgdaBound{m}\AgdaSpace{}%
\AgdaOperator{\AgdaDatatype{≤}}\AgdaSpace{}%
\AgdaBound{n}\AgdaSpace{}%
\AgdaSymbol{⦄}\AgdaSpace{}%
\AgdaSymbol{→}\AgdaSpace{}%
\AgdaOperator{\AgdaDatatype{Arg[}}\AgdaSpace{}%
\AgdaBound{n}\AgdaSpace{}%
\AgdaOperator{\AgdaDatatype{]}}\AgdaSpace{}%
\AgdaBound{ns}\<%
\\
\>[2]\AgdaFunction{Trm‿}\AgdaSpace{}%
\AgdaSymbol{(}\AgdaInductiveConstructor{𝐢}\AgdaSpace{}%
\AgdaBound{i}\AgdaSymbol{)}\AgdaSpace{}%
\AgdaBound{n}\AgdaSpace{}%
\AgdaSymbol{=}\AgdaSpace{}%
\AgdaInductiveConstructor{𝐢}\AgdaSpace{}%
\AgdaSymbol{(}\AgdaFunction{actFin}\AgdaSpace{}%
\AgdaBound{i}\AgdaSpace{}%
\AgdaBound{n}\AgdaSymbol{)}\<%
\\
\>[2]\AgdaFunction{Trm‿}\AgdaSpace{}%
\AgdaSymbol{(}\AgdaInductiveConstructor{𝐚}\AgdaSpace{}%
\AgdaBound{x}\AgdaSymbol{)}\AgdaSpace{}%
\AgdaBound{n}\AgdaSpace{}%
\AgdaSymbol{=}\AgdaSpace{}%
\AgdaInductiveConstructor{𝐚}\AgdaSpace{}%
\AgdaBound{x}\<%
\\
\>[2]\AgdaFunction{Trm‿}\AgdaSpace{}%
\AgdaSymbol{(}\AgdaInductiveConstructor{𝐨}\AgdaSpace{}%
\AgdaBound{op}\AgdaSpace{}%
\AgdaBound{ts}\AgdaSymbol{)}\AgdaSpace{}%
\AgdaBound{n}\AgdaSpace{}%
\AgdaSymbol{=}\AgdaSpace{}%
\AgdaInductiveConstructor{𝐨}\AgdaSpace{}%
\AgdaBound{op}\AgdaSpace{}%
\AgdaSymbol{(}\AgdaFunction{Arg‿}\AgdaSpace{}%
\AgdaBound{ts}\AgdaSpace{}%
\AgdaBound{n}\AgdaSymbol{)}\<%
\\
\>[2]\AgdaFunction{Arg‿}\AgdaSpace{}%
\AgdaInductiveConstructor{[]}\AgdaSpace{}%
\AgdaSymbol{\AgdaUnderscore{}}\AgdaSpace{}%
\AgdaSymbol{=}\AgdaSpace{}%
\AgdaInductiveConstructor{[]}\<%
\\
\>[2]\AgdaFunction{Arg‿}\AgdaSpace{}%
\AgdaSymbol{(}\AgdaBound{t}\AgdaSpace{}%
\AgdaOperator{\AgdaInductiveConstructor{::}}\AgdaSpace{}%
\AgdaBound{ts}\AgdaSymbol{)}\AgdaSpace{}%
\AgdaBound{n}\AgdaSpace{}%
\AgdaSymbol{=}\AgdaSpace{}%
\AgdaSymbol{(}\AgdaFunction{Trm‿}\AgdaSpace{}%
\AgdaBound{t}\AgdaSpace{}%
\AgdaSymbol{(\AgdaUnderscore{}}\AgdaSpace{}%
\AgdaOperator{\AgdaPrimitive{+}}\AgdaSpace{}%
\AgdaBound{n}\AgdaSymbol{)}\AgdaSpace{}%
\AgdaSymbol{⦃}\AgdaSpace{}%
\AgdaFunction{+≤}\AgdaSpace{}%
\AgdaSymbol{⦄)}\AgdaSpace{}%
\AgdaOperator{\AgdaInductiveConstructor{::}}\AgdaSpace{}%
\AgdaSymbol{(}\AgdaFunction{Arg‿}\AgdaSpace{}%
\AgdaBound{ts}\AgdaSpace{}%
\AgdaBound{n}\AgdaSymbol{)}\<%
\end{code}}

\newcommand{\LAMBDA}{
\begin{code}%
\>[0]\AgdaKeyword{open}\AgdaSpace{}%
\AgdaKeyword{import}\AgdaSpace{}%
\AgdaModule{WSLN}\AgdaSpace{}%
\AgdaKeyword{public}\<%
\\
\\[\AgdaEmptyExtraSkip]%
\>[0]\AgdaComment{--\ Operators}\<%
\\
\>[0]\AgdaKeyword{data}\AgdaSpace{}%
\AgdaDatatype{OpLAM}\AgdaSpace{}%
\AgdaSymbol{:}\AgdaSpace{}%
\AgdaPrimitive{Set}\AgdaSpace{}%
\AgdaKeyword{where}\<%
\\
\>[0][@{}l@{\AgdaIndent{0}}]%
\>[2]\AgdaInductiveConstructor{′lm′}\AgdaSpace{}%
\AgdaInductiveConstructor{′ap′}\AgdaSpace{}%
\AgdaSymbol{:}\AgdaSpace{}%
\AgdaDatatype{OpLAM}\<%
\\
\\[\AgdaEmptyExtraSkip]%
\>[0]\AgdaComment{--\ Arities}\<%
\\
\>[0]\AgdaFunction{arLAM}\AgdaSpace{}%
\AgdaSymbol{:}\AgdaSpace{}%
\AgdaDatatype{OpLAM}\AgdaSpace{}%
\AgdaSymbol{→}\AgdaSpace{}%
\AgdaDatatype{List}\AgdaSpace{}%
\AgdaDatatype{ℕ}\<%
\\
\>[0]\AgdaFunction{arLAM}\AgdaSpace{}%
\AgdaInductiveConstructor{′lm′}\AgdaSpace{}%
\AgdaSymbol{=}\AgdaSpace{}%
\AgdaNumber{1}\AgdaSpace{}%
\AgdaOperator{\AgdaInductiveConstructor{::}}\AgdaSpace{}%
\AgdaInductiveConstructor{[]}\<%
\\
\>[0]\AgdaFunction{arLAM}\AgdaSpace{}%
\AgdaInductiveConstructor{′ap′}\AgdaSpace{}%
\AgdaSymbol{=}\AgdaSpace{}%
\AgdaNumber{0}\AgdaSpace{}%
\AgdaOperator{\AgdaInductiveConstructor{::}}\AgdaSpace{}%
\AgdaNumber{0}\AgdaSpace{}%
\AgdaOperator{\AgdaInductiveConstructor{::}}\AgdaSpace{}%
\AgdaInductiveConstructor{[]}\<%
\\
\\[\AgdaEmptyExtraSkip]%
\>[0]\AgdaComment{--\ Binding\ signature}\<%
\\
\>[0]\AgdaKeyword{instance}\<%
\\
\>[0][@{}l@{\AgdaIndent{0}}]%
\>[2]\AgdaFunction{LAM}\AgdaSpace{}%
\AgdaSymbol{:}\AgdaSpace{}%
\AgdaRecord{Sig}\<%
\\
\>[2]\AgdaFunction{LAM}\AgdaSpace{}%
\AgdaSymbol{=}\AgdaSpace{}%
\AgdaInductiveConstructor{mkSig}\AgdaSpace{}%
\AgdaDatatype{OpLAM}\AgdaSpace{}%
\AgdaFunction{arLAM}\<%
\\
\\[\AgdaEmptyExtraSkip]%
\>[0]\AgdaComment{--\ Notation}\<%
\\
\>[0]\AgdaKeyword{infixl}\AgdaSpace{}%
\AgdaNumber{7}\AgdaSpace{}%
\AgdaOperator{\AgdaInductiveConstructor{\AgdaUnderscore{}∙\AgdaUnderscore{}}}\<%
\\
\>[0]\AgdaKeyword{pattern}\AgdaSpace{}%
\AgdaInductiveConstructor{𝐯}\AgdaSpace{}%
\AgdaBound{x}\AgdaSpace{}%
\AgdaSymbol{=}\AgdaSpace{}%
\AgdaInductiveConstructor{𝐚}\AgdaSpace{}%
\AgdaBound{x}\<%
\\
\>[0]\AgdaKeyword{pattern}\AgdaSpace{}%
\AgdaInductiveConstructor{𝛌}\AgdaSpace{}%
\AgdaBound{a}\AgdaSpace{}%
\AgdaSymbol{=}\AgdaSpace{}%
\AgdaInductiveConstructor{𝐨}\AgdaSpace{}%
\AgdaInductiveConstructor{′lm′}\AgdaSpace{}%
\AgdaSymbol{(}\AgdaBound{a}\AgdaSpace{}%
\AgdaOperator{\AgdaInductiveConstructor{::}}\AgdaSpace{}%
\AgdaInductiveConstructor{[]}\AgdaSymbol{)}\<%
\\
\>[0]\AgdaKeyword{pattern}\AgdaSpace{}%
\AgdaOperator{\AgdaInductiveConstructor{\AgdaUnderscore{}∙\AgdaUnderscore{}}}\AgdaSpace{}%
\AgdaBound{b}\AgdaSpace{}%
\AgdaBound{a}\AgdaSpace{}%
\AgdaSymbol{=}\AgdaSpace{}%
\AgdaInductiveConstructor{𝐨}\AgdaSpace{}%
\AgdaInductiveConstructor{′ap′}\AgdaSpace{}%
\AgdaSymbol{(}\AgdaBound{b}\AgdaSpace{}%
\AgdaOperator{\AgdaInductiveConstructor{::}}\AgdaSpace{}%
\AgdaBound{a}\AgdaSpace{}%
\AgdaOperator{\AgdaInductiveConstructor{::}}\AgdaSpace{}%
\AgdaInductiveConstructor{[]}\AgdaSymbol{)}\<%
\end{code}}

\newcommand{\LN}{
\begin{code}%
\>[0]\AgdaKeyword{data}\AgdaSpace{}%
\AgdaDatatype{Λ}\AgdaSpace{}%
\AgdaSymbol{:}\AgdaSpace{}%
\AgdaPrimitive{Set}\AgdaSpace{}%
\AgdaKeyword{where}\<%
\\
\>[0][@{}l@{\AgdaIndent{0}}]%
\>[2]\AgdaInductiveConstructor{a}\AgdaSpace{}%
\AgdaSymbol{:}\AgdaSpace{}%
\AgdaFunction{𝔸}\AgdaSpace{}%
\AgdaSymbol{→}\AgdaSpace{}%
\AgdaDatatype{Λ}\<%
\\
\>[2]\AgdaInductiveConstructor{i}\AgdaSpace{}%
\AgdaSymbol{:}\AgdaSpace{}%
\AgdaDatatype{ℕ}\AgdaSpace{}%
\AgdaSymbol{→}\AgdaSpace{}%
\AgdaDatatype{Λ}\<%
\\
\>[2]\AgdaInductiveConstructor{𝛌}\AgdaSpace{}%
\AgdaSymbol{:}\AgdaSpace{}%
\AgdaDatatype{Λ}\AgdaSpace{}%
\AgdaSymbol{→}\AgdaSpace{}%
\AgdaDatatype{Λ}\<%
\\
\>[2]\AgdaOperator{\AgdaInductiveConstructor{\AgdaUnderscore{}∙\AgdaUnderscore{}}}\AgdaSpace{}%
\AgdaSymbol{:}\AgdaSpace{}%
\AgdaDatatype{Λ}\AgdaSpace{}%
\AgdaSymbol{→}\AgdaSpace{}%
\AgdaDatatype{Λ}\AgdaSpace{}%
\AgdaSymbol{→}\AgdaSpace{}%
\AgdaDatatype{Λ}\<%
\end{code}}

\newcommand{\PiC}{
\begin{code}%
\>[0]\AgdaKeyword{open}\AgdaSpace{}%
\AgdaKeyword{import}\AgdaSpace{}%
\AgdaModule{WSLN}\AgdaSpace{}%
\AgdaKeyword{public}\<%
\\
\>[0]\AgdaKeyword{instance}\<%
\\
\>[0][@{}l@{\AgdaIndent{0}}]%
\>[2]\AgdaFunction{PiC}\AgdaSpace{}%
\AgdaSymbol{:}\AgdaSpace{}%
\AgdaRecord{Sig}\<%
\\
\>[2]\AgdaFunction{PiC}\AgdaSpace{}%
\AgdaSymbol{=}\AgdaSpace{}%
\AgdaInductiveConstructor{mkSig}\AgdaSpace{}%
\AgdaDatatype{OpPiC}\AgdaSpace{}%
\AgdaFunction{arPiC}\<%
\\
\>[2][@{}l@{\AgdaIndent{0}}]%
\>[4]\AgdaKeyword{module}\AgdaSpace{}%
\AgdaModule{\AgdaUnderscore{}}\AgdaSpace{}%
\AgdaKeyword{where}\<%
\\
\>[4]\AgdaComment{--\ Operators}\<%
\\
\>[4]\AgdaKeyword{data}\AgdaSpace{}%
\AgdaDatatype{OpPiC}\AgdaSpace{}%
\AgdaSymbol{:}\AgdaSpace{}%
\AgdaPrimitive{Set}\AgdaSpace{}%
\AgdaKeyword{where}\<%
\\
\>[4][@{}l@{\AgdaIndent{0}}]%
\>[6]\AgdaComment{--\ input\ prefixed\ process}\<%
\\
\>[6]\AgdaInductiveConstructor{′in′}\AgdaSpace{}%
\AgdaSymbol{:}\AgdaSpace{}%
\AgdaDatatype{OpPiC}\<%
\\
\>[6]\AgdaComment{--\ output\ prefixed\ process}\<%
\\
\>[6]\AgdaInductiveConstructor{′out′}\AgdaSpace{}%
\AgdaSymbol{:}%
\>[15]\AgdaDatatype{OpPiC}\<%
\\
\>[6]\AgdaComment{--\ \ parallel\ composition}\<%
\\
\>[6]\AgdaInductiveConstructor{′par′}\AgdaSpace{}%
\AgdaSymbol{:}\AgdaSpace{}%
\AgdaDatatype{OpPiC}\<%
\\
\>[6]\AgdaComment{--\ \ restriction}\<%
\\
\>[6]\AgdaInductiveConstructor{′nu′}\AgdaSpace{}%
\AgdaSymbol{:}\AgdaSpace{}%
\AgdaDatatype{OpPiC}\<%
\\
\>[6]\AgdaComment{--\ replication}\<%
\\
\>[6]\AgdaInductiveConstructor{′repl′}\AgdaSpace{}%
\AgdaSymbol{:}\AgdaSpace{}%
\AgdaDatatype{OpPiC}\<%
\\
\>[6]\AgdaComment{--\ termination}\<%
\\
\>[6]\AgdaInductiveConstructor{′null′}\AgdaSpace{}%
\AgdaSymbol{:}\AgdaSpace{}%
\AgdaDatatype{OpPiC}\<%
\\
\\[\AgdaEmptyExtraSkip]%
\>[4]\AgdaComment{--\ Arities}\<%
\\
\>[4]\AgdaFunction{arPiC}\AgdaSpace{}%
\AgdaSymbol{:}\AgdaSpace{}%
\AgdaDatatype{OpPiC}\AgdaSpace{}%
\AgdaSymbol{→}\AgdaSpace{}%
\AgdaDatatype{List}\AgdaSpace{}%
\AgdaDatatype{ℕ}\<%
\\
\>[4]\AgdaFunction{arPiC}\AgdaSpace{}%
\AgdaInductiveConstructor{′in′}\AgdaSpace{}%
\AgdaSymbol{=}\AgdaSpace{}%
\AgdaNumber{0}\AgdaSpace{}%
\AgdaOperator{\AgdaInductiveConstructor{::}}\AgdaSpace{}%
\AgdaNumber{1}\AgdaSpace{}%
\AgdaOperator{\AgdaInductiveConstructor{::}}\AgdaSpace{}%
\AgdaInductiveConstructor{[]}\<%
\\
\>[4]\AgdaFunction{arPiC}\AgdaSpace{}%
\AgdaInductiveConstructor{′out′}\AgdaSpace{}%
\AgdaSymbol{=}\AgdaSpace{}%
\AgdaNumber{0}\AgdaSpace{}%
\AgdaOperator{\AgdaInductiveConstructor{::}}\AgdaSpace{}%
\AgdaNumber{0}\AgdaSpace{}%
\AgdaOperator{\AgdaInductiveConstructor{::}}\AgdaSpace{}%
\AgdaNumber{0}\AgdaSpace{}%
\AgdaOperator{\AgdaInductiveConstructor{::}}\AgdaSpace{}%
\AgdaInductiveConstructor{[]}\<%
\\
\>[4]\AgdaFunction{arPiC}\AgdaSpace{}%
\AgdaInductiveConstructor{′par′}\AgdaSpace{}%
\AgdaSymbol{=}\AgdaSpace{}%
\AgdaNumber{0}\AgdaSpace{}%
\AgdaOperator{\AgdaInductiveConstructor{::}}\AgdaSpace{}%
\AgdaNumber{0}\AgdaSpace{}%
\AgdaOperator{\AgdaInductiveConstructor{::}}\AgdaSpace{}%
\AgdaInductiveConstructor{[]}\<%
\\
\>[4]\AgdaFunction{arPiC}\AgdaSpace{}%
\AgdaInductiveConstructor{′nu′}\AgdaSpace{}%
\AgdaSymbol{=}\AgdaSpace{}%
\AgdaNumber{1}\AgdaSpace{}%
\AgdaOperator{\AgdaInductiveConstructor{::}}\AgdaSpace{}%
\AgdaInductiveConstructor{[]}\<%
\\
\>[4]\AgdaFunction{arPiC}\AgdaSpace{}%
\AgdaInductiveConstructor{′repl′}\AgdaSpace{}%
\AgdaSymbol{=}\AgdaSpace{}%
\AgdaNumber{0}\AgdaSpace{}%
\AgdaOperator{\AgdaInductiveConstructor{::}}\AgdaSpace{}%
\AgdaInductiveConstructor{[]}\<%
\\
\>[4]\AgdaFunction{arPiC}\AgdaSpace{}%
\AgdaInductiveConstructor{′null′}\AgdaSpace{}%
\AgdaSymbol{=}\AgdaSpace{}%
\AgdaInductiveConstructor{[]}\<%
\\
\\[\AgdaEmptyExtraSkip]%
\>[0]\AgdaComment{--\ Notation}\<%
\\
\>[0]\AgdaKeyword{infixl}\AgdaSpace{}%
\AgdaNumber{5}\AgdaSpace{}%
\AgdaOperator{\AgdaInductiveConstructor{\AgdaUnderscore{}∣\AgdaUnderscore{}}}\<%
\\
\>[0]\AgdaKeyword{pattern}\AgdaSpace{}%
\AgdaInductiveConstructor{𝐢𝐧}\AgdaSpace{}%
\AgdaBound{a}\AgdaSpace{}%
\AgdaBound{P}\AgdaSpace{}%
\AgdaSymbol{=}\AgdaSpace{}%
\AgdaInductiveConstructor{𝐨}\AgdaSpace{}%
\AgdaInductiveConstructor{′in′}\AgdaSpace{}%
\AgdaSymbol{(}\AgdaInductiveConstructor{𝐚}\AgdaSpace{}%
\AgdaBound{a}\AgdaSpace{}%
\AgdaOperator{\AgdaInductiveConstructor{::}}\AgdaSpace{}%
\AgdaBound{P}\AgdaSpace{}%
\AgdaOperator{\AgdaInductiveConstructor{::}}\AgdaSpace{}%
\AgdaInductiveConstructor{[]}\AgdaSymbol{)}\<%
\\
\>[0]\AgdaKeyword{pattern}\AgdaSpace{}%
\AgdaInductiveConstructor{𝐨𝐮𝐭}\AgdaSpace{}%
\AgdaBound{a}\AgdaSpace{}%
\AgdaBound{b}\AgdaSpace{}%
\AgdaBound{P}\AgdaSpace{}%
\AgdaSymbol{=}\AgdaSpace{}%
\AgdaInductiveConstructor{𝐨}\AgdaSpace{}%
\AgdaInductiveConstructor{′out′}\AgdaSpace{}%
\AgdaSymbol{(}\AgdaInductiveConstructor{𝐚}\AgdaSpace{}%
\AgdaBound{a}\AgdaSpace{}%
\AgdaOperator{\AgdaInductiveConstructor{::}}\AgdaSpace{}%
\AgdaInductiveConstructor{𝐚}\AgdaSpace{}%
\AgdaBound{b}\AgdaSpace{}%
\AgdaOperator{\AgdaInductiveConstructor{::}}\AgdaSpace{}%
\AgdaBound{P}\AgdaSpace{}%
\AgdaOperator{\AgdaInductiveConstructor{::}}\AgdaSpace{}%
\AgdaInductiveConstructor{[]}\AgdaSymbol{)}\<%
\\
\>[0]\AgdaKeyword{pattern}\AgdaSpace{}%
\AgdaOperator{\AgdaInductiveConstructor{\AgdaUnderscore{}∣\AgdaUnderscore{}}}\AgdaSpace{}%
\AgdaBound{P}\AgdaSpace{}%
\AgdaBound{Q}\AgdaSpace{}%
\AgdaSymbol{=}\AgdaSpace{}%
\AgdaInductiveConstructor{𝐨}\AgdaSpace{}%
\AgdaInductiveConstructor{′par′}\AgdaSpace{}%
\AgdaSymbol{(}\AgdaBound{P}\AgdaSpace{}%
\AgdaOperator{\AgdaInductiveConstructor{::}}\AgdaSpace{}%
\AgdaBound{Q}\AgdaSpace{}%
\AgdaOperator{\AgdaInductiveConstructor{::}}\AgdaSpace{}%
\AgdaInductiveConstructor{[]}\AgdaSymbol{)}\<%
\\
\>[0]\AgdaKeyword{pattern}\AgdaSpace{}%
\AgdaInductiveConstructor{ν}\AgdaSpace{}%
\AgdaBound{P}\AgdaSpace{}%
\AgdaSymbol{=}\AgdaSpace{}%
\AgdaInductiveConstructor{𝐨}\AgdaSpace{}%
\AgdaInductiveConstructor{′nu′}\AgdaSpace{}%
\AgdaSymbol{(}\AgdaBound{P}\AgdaSpace{}%
\AgdaOperator{\AgdaInductiveConstructor{::}}\AgdaSpace{}%
\AgdaInductiveConstructor{[]}\AgdaSymbol{)}\<%
\\
\>[0]\AgdaKeyword{pattern}\AgdaSpace{}%
\AgdaInductiveConstructor{‼}\AgdaSpace{}%
\AgdaBound{P}\AgdaSpace{}%
\AgdaSymbol{=}\AgdaSpace{}%
\AgdaInductiveConstructor{𝐨}\AgdaSpace{}%
\AgdaInductiveConstructor{′repl′}\AgdaSpace{}%
\AgdaSymbol{(}\AgdaBound{P}\AgdaSpace{}%
\AgdaOperator{\AgdaInductiveConstructor{::}}\AgdaSpace{}%
\AgdaInductiveConstructor{[]}\AgdaSymbol{)}\<%
\\
\>[0]\AgdaKeyword{pattern}\AgdaSpace{}%
\AgdaInductiveConstructor{𝐎}\AgdaSpace{}%
\AgdaSymbol{=}\AgdaSpace{}%
\AgdaInductiveConstructor{𝐨}\AgdaSpace{}%
\AgdaInductiveConstructor{′null′}\AgdaSpace{}%
\AgdaInductiveConstructor{[]}\<%
\end{code}}

\newcommand{\Eg}{
\begin{code}%
\>[0][@{}l@{\AgdaIndent{1}}]%
\>[2]\AgdaInductiveConstructor{𝐨}\AgdaSpace{}%
\AgdaInductiveConstructor{′out′}\AgdaSpace{}%
\AgdaSymbol{((}\AgdaInductiveConstructor{𝐨}\AgdaSpace{}%
\AgdaInductiveConstructor{′null′}\AgdaSpace{}%
\AgdaInductiveConstructor{[]}\AgdaSymbol{)}\AgdaSpace{}%
\AgdaOperator{\AgdaInductiveConstructor{::}}\AgdaSpace{}%
\AgdaSymbol{(}\AgdaInductiveConstructor{𝐨}\AgdaSpace{}%
\AgdaInductiveConstructor{′null′}\AgdaSpace{}%
\AgdaInductiveConstructor{[]}\AgdaSymbol{)}\AgdaSpace{}%
\AgdaOperator{\AgdaInductiveConstructor{::}}\AgdaSpace{}%
\AgdaSymbol{(}\AgdaInductiveConstructor{𝐨}\AgdaSpace{}%
\AgdaInductiveConstructor{′null′}\AgdaSpace{}%
\AgdaInductiveConstructor{[]}\AgdaSymbol{)}\AgdaSpace{}%
\AgdaOperator{\AgdaInductiveConstructor{::}}\AgdaSpace{}%
\AgdaInductiveConstructor{[]}\AgdaSymbol{)}\<%
\end{code}}

\newcommand{\ARITY}{
\begin{code}%
\>[0][@{}l@{\AgdaIndent{1}}]%
\>[2]\AgdaNumber{1}\AgdaSpace{}%
\AgdaOperator{\AgdaInductiveConstructor{::}}\AgdaSpace{}%
\AgdaNumber{0}\AgdaSpace{}%
\AgdaOperator{\AgdaInductiveConstructor{::}}\AgdaSpace{}%
\AgdaNumber{2}\AgdaSpace{}%
\AgdaOperator{\AgdaInductiveConstructor{::}}\AgdaSpace{}%
\AgdaNumber{0}\AgdaSpace{}%
\AgdaOperator{\AgdaInductiveConstructor{::}}\AgdaSpace{}%
\AgdaInductiveConstructor{[]}\<%
\end{code}}

\newcommand{\PATTERN}{
\begin{code}%
\>[0]\AgdaKeyword{pattern}\AgdaSpace{}%
\AgdaInductiveConstructor{𝐧𝐫𝐞𝐜}\AgdaSpace{}%
\AgdaBound{C}\AgdaSpace{}%
\AgdaBound{c}\AgdaSpace{}%
\AgdaBound{d}\AgdaSpace{}%
\AgdaBound{a}\AgdaSpace{}%
\AgdaSymbol{=}\AgdaSpace{}%
\AgdaInductiveConstructor{𝐨}\AgdaSpace{}%
\AgdaInductiveConstructor{′natrec′}\AgdaSpace{}%
\AgdaSymbol{(}\AgdaBound{C}\AgdaSpace{}%
\AgdaOperator{\AgdaInductiveConstructor{::}}\AgdaSpace{}%
\AgdaBound{c}\AgdaSpace{}%
\AgdaOperator{\AgdaInductiveConstructor{::}}\AgdaSpace{}%
\AgdaBound{d}\AgdaSpace{}%
\AgdaOperator{\AgdaInductiveConstructor{::}}\AgdaSpace{}%
\AgdaBound{a}\AgdaSpace{}%
\AgdaOperator{\AgdaInductiveConstructor{::}}\AgdaSpace{}%
\AgdaInductiveConstructor{[]}\AgdaSymbol{)}\<%
\end{code}}

\newcommand{\NELIM}{
\begin{code}%
\>[0][@{}l@{\AgdaIndent{1}}]%
\>[2]\AgdaInductiveConstructor{⊢𝐧𝐫𝐞𝐜⁻}\AgdaSpace{}%
\AgdaSymbol{:}\<%
\\
\>[2][@{}l@{\AgdaIndent{0}}]%
\>[4]\AgdaSymbol{\{}\AgdaBound{l}\AgdaSpace{}%
\AgdaSymbol{:}\AgdaSpace{}%
\AgdaDatatype{ℕ}\AgdaSymbol{\}}\<%
\\
\>[4]\AgdaSymbol{\{}\AgdaBound{C}\AgdaSpace{}%
\AgdaSymbol{:}\AgdaSpace{}%
\AgdaOperator{\AgdaDatatype{Trm[}}\AgdaSpace{}%
\AgdaNumber{1}\AgdaSpace{}%
\AgdaOperator{\AgdaDatatype{]}}\AgdaSymbol{\}}\<%
\\
\>[4]\AgdaSymbol{\{}\AgdaBound{c₀}\AgdaSpace{}%
\AgdaBound{a}\AgdaSpace{}%
\AgdaSymbol{:}\AgdaSpace{}%
\AgdaFunction{Trm}\AgdaSymbol{\}}\<%
\\
\>[4]\AgdaSymbol{\{}\AgdaBound{c₊}\AgdaSpace{}%
\AgdaSymbol{:}\AgdaSpace{}%
\AgdaOperator{\AgdaDatatype{Trm[}}\AgdaSpace{}%
\AgdaNumber{2}\AgdaSpace{}%
\AgdaOperator{\AgdaDatatype{]}}\AgdaSymbol{\}}\<%
\\
\>[4]\AgdaSymbol{\{}\AgdaBound{x}\AgdaSpace{}%
\AgdaBound{y}\AgdaSpace{}%
\AgdaSymbol{:}\AgdaSpace{}%
\AgdaFunction{𝔸}\AgdaSymbol{\}}\<%
\\
\>[4]\AgdaSymbol{(}\AgdaBound{\AgdaUnderscore{}}\AgdaSpace{}%
\AgdaSymbol{:}\AgdaSpace{}%
\AgdaBound{Γ}\AgdaSpace{}%
\AgdaOperator{\AgdaDatatype{⊢}}\AgdaSpace{}%
\AgdaBound{c₀}\AgdaSpace{}%
\AgdaOperator{\AgdaInductiveConstructor{∶}}\AgdaSpace{}%
\AgdaBound{C}\AgdaSpace{}%
\AgdaOperator{\AgdaField{[}}\AgdaSpace{}%
\AgdaInductiveConstructor{𝐳𝐞𝐫𝐨}\AgdaSpace{}%
\AgdaOperator{\AgdaField{]}}\AgdaSpace{}%
\AgdaOperator{\AgdaInductiveConstructor{⦂}}\AgdaSpace{}%
\AgdaBound{l}\AgdaSymbol{)}\<%
\\
\>[4]\AgdaSymbol{(}\AgdaBound{\AgdaUnderscore{}}\AgdaSpace{}%
\AgdaSymbol{:}\AgdaSpace{}%
\AgdaSymbol{(}\AgdaBound{Γ}\AgdaSpace{}%
\AgdaOperator{\AgdaInductiveConstructor{⨟}}\AgdaSpace{}%
\AgdaBound{x}\AgdaSpace{}%
\AgdaOperator{\AgdaInductiveConstructor{∶}}\AgdaSpace{}%
\AgdaInductiveConstructor{𝐍𝐚𝐭}\AgdaSpace{}%
\AgdaOperator{\AgdaInductiveConstructor{⦂}}\AgdaSpace{}%
\AgdaNumber{0}\AgdaSpace{}%
\AgdaOperator{\AgdaInductiveConstructor{⨟}}\AgdaSpace{}%
\AgdaBound{y}\AgdaSpace{}%
\AgdaOperator{\AgdaInductiveConstructor{∶}}\AgdaSpace{}%
\AgdaBound{C}\AgdaSpace{}%
\AgdaOperator{\AgdaField{[}}\AgdaSpace{}%
\AgdaInductiveConstructor{𝐯}\AgdaSpace{}%
\AgdaBound{x}\AgdaSpace{}%
\AgdaOperator{\AgdaField{]}}\AgdaSpace{}%
\AgdaOperator{\AgdaInductiveConstructor{⦂}}\AgdaSpace{}%
\AgdaBound{l}\AgdaSymbol{)}\AgdaSpace{}%
\AgdaOperator{\AgdaDatatype{⊢}}\<%
\\
\>[4][@{}l@{\AgdaIndent{0}}]%
\>[6]\AgdaBound{c₊}\AgdaSpace{}%
\AgdaOperator{\AgdaFunction{[}}\AgdaSpace{}%
\AgdaInductiveConstructor{𝐯}\AgdaSpace{}%
\AgdaBound{x}\AgdaSpace{}%
\AgdaOperator{\AgdaFunction{][}}\AgdaSpace{}%
\AgdaInductiveConstructor{𝐯}\AgdaSpace{}%
\AgdaBound{y}\AgdaSpace{}%
\AgdaOperator{\AgdaFunction{]}}\AgdaSpace{}%
\AgdaOperator{\AgdaInductiveConstructor{∶}}\AgdaSpace{}%
\AgdaBound{C}\AgdaSpace{}%
\AgdaOperator{\AgdaField{[}}\AgdaSpace{}%
\AgdaInductiveConstructor{𝐬𝐮𝐜𝐜}\AgdaSpace{}%
\AgdaSymbol{(}\AgdaInductiveConstructor{𝐯}\AgdaSpace{}%
\AgdaBound{x}\AgdaSymbol{)}\AgdaSpace{}%
\AgdaOperator{\AgdaField{]}}\AgdaSpace{}%
\AgdaOperator{\AgdaInductiveConstructor{⦂}}\AgdaSpace{}%
\AgdaBound{l}\AgdaSymbol{)}\<%
\\
\>[4]\AgdaSymbol{(}\AgdaBound{\AgdaUnderscore{}}\AgdaSpace{}%
\AgdaSymbol{:}\AgdaSpace{}%
\AgdaBound{Γ}\AgdaSpace{}%
\AgdaOperator{\AgdaDatatype{⊢}}\AgdaSpace{}%
\AgdaBound{a}\AgdaSpace{}%
\AgdaOperator{\AgdaInductiveConstructor{∶}}\AgdaSpace{}%
\AgdaInductiveConstructor{𝐍𝐚𝐭}\AgdaSpace{}%
\AgdaOperator{\AgdaInductiveConstructor{⦂}}\AgdaSpace{}%
\AgdaNumber{0}\AgdaSymbol{)}\<%
\\
\>[4]\AgdaSymbol{(}\AgdaBound{\AgdaUnderscore{}}\AgdaSpace{}%
\AgdaSymbol{:}\AgdaSpace{}%
\AgdaBound{x}\AgdaSpace{}%
\AgdaOperator{\AgdaFunction{\#}}\AgdaSpace{}%
\AgdaSymbol{(}\AgdaBound{C}\AgdaSpace{}%
\AgdaOperator{\AgdaInductiveConstructor{,}}\AgdaSpace{}%
\AgdaBound{c₊}\AgdaSymbol{))}\<%
\\
\>[4]\AgdaSymbol{(}\AgdaBound{\AgdaUnderscore{}}\AgdaSpace{}%
\AgdaSymbol{:}\AgdaSpace{}%
\AgdaBound{y}\AgdaSpace{}%
\AgdaOperator{\AgdaFunction{\#}}\AgdaSpace{}%
\AgdaBound{c₊}\AgdaSymbol{)}\<%
\\
\>[4]\AgdaSymbol{→}\AgdaSpace{}%
\AgdaComment{----------------------------------}\<%
\\
\>[4]\AgdaBound{Γ}\AgdaSpace{}%
\AgdaOperator{\AgdaDatatype{⊢}}\AgdaSpace{}%
\AgdaInductiveConstructor{𝐧𝐫𝐞𝐜}\AgdaSpace{}%
\AgdaBound{C}\AgdaSpace{}%
\AgdaBound{c₀}\AgdaSpace{}%
\AgdaBound{c₊}\AgdaSpace{}%
\AgdaBound{a}\AgdaSpace{}%
\AgdaOperator{\AgdaInductiveConstructor{∶}}\AgdaSpace{}%
\AgdaBound{C}\AgdaSpace{}%
\AgdaOperator{\AgdaField{[}}\AgdaSpace{}%
\AgdaBound{a}\AgdaSpace{}%
\AgdaOperator{\AgdaField{]}}\AgdaSpace{}%
\AgdaOperator{\AgdaInductiveConstructor{⦂}}\AgdaSpace{}%
\AgdaBound{l}\<%
\end{code}}

\newcommand{\NREC}{
\begin{code}%
\>[2]\AgdaInductiveConstructor{⊢𝐧𝐫𝐞𝐜}\AgdaSpace{}%
\AgdaSymbol{:}\<%
\\
\>[2][@{}l@{\AgdaIndent{0}}]%
\>[4]\AgdaSymbol{\{}\AgdaBound{l}\AgdaSpace{}%
\AgdaSymbol{:}\AgdaSpace{}%
\AgdaDatatype{ℕ}\AgdaSymbol{\}}\<%
\\
\>[4]\AgdaSymbol{\{}\AgdaBound{C}\AgdaSpace{}%
\AgdaSymbol{:}\AgdaSpace{}%
\AgdaOperator{\AgdaDatatype{Trm[}}\AgdaSpace{}%
\AgdaNumber{1}\AgdaSpace{}%
\AgdaOperator{\AgdaDatatype{]}}\AgdaSymbol{\}}\<%
\\
\>[4]\AgdaSymbol{\{}\AgdaBound{c₀}\AgdaSpace{}%
\AgdaBound{a}\AgdaSpace{}%
\AgdaSymbol{:}\AgdaSpace{}%
\AgdaFunction{Trm}\AgdaSymbol{\}}\<%
\\
\>[4]\AgdaSymbol{\{}\AgdaBound{c₊}\AgdaSpace{}%
\AgdaSymbol{:}\AgdaSpace{}%
\AgdaOperator{\AgdaDatatype{Trm[}}\AgdaSpace{}%
\AgdaNumber{2}\AgdaSpace{}%
\AgdaOperator{\AgdaDatatype{]}}\AgdaSymbol{\}}\<%
\\
\>[4]\AgdaSymbol{(}\AgdaBound{S}\AgdaSpace{}%
\AgdaSymbol{:}\AgdaSpace{}%
\AgdaDatatype{Fset𝔸}\AgdaSymbol{)}\<%
\\
\>[4]\AgdaSymbol{(}\AgdaBound{\AgdaUnderscore{}}\AgdaSpace{}%
\AgdaSymbol{:}\AgdaSpace{}%
\AgdaBound{Γ}\AgdaSpace{}%
\AgdaOperator{\AgdaDatatype{⊢}}\AgdaSpace{}%
\AgdaBound{c₀}\AgdaSpace{}%
\AgdaOperator{\AgdaInductiveConstructor{∶}}\AgdaSpace{}%
\AgdaBound{C}\AgdaSpace{}%
\AgdaOperator{\AgdaField{[}}\AgdaSpace{}%
\AgdaInductiveConstructor{𝐳𝐞𝐫𝐨}\AgdaSpace{}%
\AgdaOperator{\AgdaField{]}}\AgdaSpace{}%
\AgdaOperator{\AgdaInductiveConstructor{⦂}}\AgdaSpace{}%
\AgdaBound{l}\AgdaSymbol{)}\<%
\\
\>[4]\AgdaSymbol{(}\AgdaBound{\AgdaUnderscore{}}\AgdaSpace{}%
\AgdaSymbol{:}\AgdaSpace{}%
\AgdaSymbol{∀}\AgdaSpace{}%
\AgdaBound{x}\AgdaSpace{}%
\AgdaBound{y}\AgdaSpace{}%
\AgdaSymbol{→}\AgdaSpace{}%
\AgdaBound{x}\AgdaSpace{}%
\AgdaOperator{\AgdaFunction{\#}}\AgdaSpace{}%
\AgdaBound{y}\AgdaSpace{}%
\AgdaOperator{\AgdaFunction{\#}}\AgdaSpace{}%
\AgdaBound{S}\AgdaSpace{}%
\AgdaSymbol{→}\<%
\\
\>[4][@{}l@{\AgdaIndent{0}}]%
\>[6]\AgdaSymbol{(}\AgdaBound{Γ}\AgdaSpace{}%
\AgdaOperator{\AgdaInductiveConstructor{⨟}}\AgdaSpace{}%
\AgdaBound{x}\AgdaSpace{}%
\AgdaOperator{\AgdaInductiveConstructor{∶}}\AgdaSpace{}%
\AgdaInductiveConstructor{𝐍𝐚𝐭}\AgdaSpace{}%
\AgdaOperator{\AgdaInductiveConstructor{⦂}}\AgdaSpace{}%
\AgdaNumber{0}\AgdaSpace{}%
\AgdaOperator{\AgdaInductiveConstructor{⨟}}\AgdaSpace{}%
\AgdaBound{y}\AgdaSpace{}%
\AgdaOperator{\AgdaInductiveConstructor{∶}}\AgdaSpace{}%
\AgdaBound{C}\AgdaSpace{}%
\AgdaOperator{\AgdaField{[}}\AgdaSpace{}%
\AgdaInductiveConstructor{𝐯}\AgdaSpace{}%
\AgdaBound{x}\AgdaSpace{}%
\AgdaOperator{\AgdaField{]}}\AgdaSpace{}%
\AgdaOperator{\AgdaInductiveConstructor{⦂}}\AgdaSpace{}%
\AgdaBound{l}\AgdaSymbol{)}\AgdaSpace{}%
\AgdaOperator{\AgdaDatatype{⊢}}\<%
\\
\>[6]\AgdaBound{c₊}\AgdaSpace{}%
\AgdaOperator{\AgdaFunction{[}}\AgdaSpace{}%
\AgdaInductiveConstructor{𝐯}\AgdaSpace{}%
\AgdaBound{x}\AgdaSpace{}%
\AgdaOperator{\AgdaFunction{][}}\AgdaSpace{}%
\AgdaInductiveConstructor{𝐯}\AgdaSpace{}%
\AgdaBound{y}\AgdaSpace{}%
\AgdaOperator{\AgdaFunction{]}}\AgdaSpace{}%
\AgdaOperator{\AgdaInductiveConstructor{∶}}\AgdaSpace{}%
\AgdaBound{C}\AgdaSpace{}%
\AgdaOperator{\AgdaField{[}}\AgdaSpace{}%
\AgdaInductiveConstructor{𝐬𝐮𝐜𝐜}\AgdaSpace{}%
\AgdaSymbol{(}\AgdaInductiveConstructor{𝐯}\AgdaSpace{}%
\AgdaBound{x}\AgdaSymbol{)}\AgdaSpace{}%
\AgdaOperator{\AgdaField{]}}\AgdaSpace{}%
\AgdaOperator{\AgdaInductiveConstructor{⦂}}\AgdaSpace{}%
\AgdaBound{l}\AgdaSymbol{)}\<%
\\
\>[4]\AgdaSymbol{(}\AgdaBound{\AgdaUnderscore{}}\AgdaSpace{}%
\AgdaSymbol{:}\AgdaSpace{}%
\AgdaBound{Γ}\AgdaSpace{}%
\AgdaOperator{\AgdaDatatype{⊢}}\AgdaSpace{}%
\AgdaBound{a}\AgdaSpace{}%
\AgdaOperator{\AgdaInductiveConstructor{∶}}\AgdaSpace{}%
\AgdaInductiveConstructor{𝐍𝐚𝐭}\AgdaSpace{}%
\AgdaOperator{\AgdaInductiveConstructor{⦂}}\AgdaSpace{}%
\AgdaNumber{0}\AgdaSymbol{)}\<%
\\
\>[4]\AgdaSymbol{(}\AgdaBound{\AgdaUnderscore{}}\AgdaSpace{}%
\AgdaSymbol{:}\AgdaSpace{}%
\AgdaSymbol{∀}\AgdaSpace{}%
\AgdaBound{x}\AgdaSpace{}%
\AgdaSymbol{→}\AgdaSpace{}%
\AgdaBound{x}\AgdaSpace{}%
\AgdaOperator{\AgdaFunction{\#}}\AgdaSpace{}%
\AgdaBound{S}\AgdaSpace{}%
\AgdaSymbol{→}\<%
\\
\>[4][@{}l@{\AgdaIndent{0}}]%
\>[6]\AgdaSymbol{(}\AgdaBound{Γ}\AgdaSpace{}%
\AgdaOperator{\AgdaInductiveConstructor{⨟}}\AgdaSpace{}%
\AgdaBound{x}\AgdaSpace{}%
\AgdaOperator{\AgdaInductiveConstructor{∶}}\AgdaSpace{}%
\AgdaInductiveConstructor{𝐍𝐚𝐭}\AgdaSpace{}%
\AgdaOperator{\AgdaInductiveConstructor{⦂}}\AgdaSpace{}%
\AgdaNumber{0}\AgdaSymbol{)}\AgdaSpace{}%
\AgdaOperator{\AgdaDatatype{⊢}}\AgdaSpace{}%
\AgdaBound{C}\AgdaSpace{}%
\AgdaOperator{\AgdaField{[}}\AgdaSpace{}%
\AgdaInductiveConstructor{𝐯}\AgdaSpace{}%
\AgdaBound{x}\AgdaSpace{}%
\AgdaOperator{\AgdaField{]}}\AgdaSpace{}%
\AgdaOperator{\AgdaFunction{⦂}}\AgdaSpace{}%
\AgdaBound{l}\AgdaSymbol{)}\AgdaSpace{}%
\AgdaComment{--\ helper\ hypothesis}\<%
\\
\>[4]\AgdaSymbol{→}\AgdaSpace{}%
\AgdaComment{----------------------------------}\<%
\\
\>[4]\AgdaBound{Γ}\AgdaSpace{}%
\AgdaOperator{\AgdaDatatype{⊢}}\AgdaSpace{}%
\AgdaInductiveConstructor{𝐧𝐫𝐞𝐜}\AgdaSpace{}%
\AgdaBound{C}\AgdaSpace{}%
\AgdaBound{c₀}\AgdaSpace{}%
\AgdaBound{c₊}\AgdaSpace{}%
\AgdaBound{a}\AgdaSpace{}%
\AgdaOperator{\AgdaInductiveConstructor{∶}}\AgdaSpace{}%
\AgdaBound{C}\AgdaSpace{}%
\AgdaOperator{\AgdaField{[}}\AgdaSpace{}%
\AgdaBound{a}\AgdaSpace{}%
\AgdaOperator{\AgdaField{]}}\AgdaSpace{}%
\AgdaOperator{\AgdaInductiveConstructor{⦂}}\AgdaSpace{}%
\AgdaBound{l}\<%
\end{code}}

\begin{document}

\title[Well-Scoped Locally Nameless Representation of
Syntax]{Well-Scoped Locally Nameless Representation of Syntax}

\author*{\fnm{Andrew M.} \sur{Pitts}}
\email{andrew.pitts@cl.cam.ac.uk}

\affil{\orgdiv{Department of Computer Science \& Technology},
  \orgname{University of Cambridge}, \city{Cambridge}, \postcode{CB3
    0FD}, \country{UK}}

\abstract{When using interactive theorem provers based on dependent
  type theory to define and reason about languages involving binding
  constructs, we advocate the use of a well-scoped version of the
  locally nameless method of representing syntax. This paper describes
  generic code for this style of syntax representation within the Agda
  theorem prover, parameterized by a notion of binding signature due
  to Aczel and Plotkin. It gives a proof of its adequacy with respect
  to na{\"i}ve nameful syntax modulo $\alpha$-conversion and discusses
  some examples of its use.}

\keywords{binding, locally nameless, well-scoped, Agda}

\maketitle
%%%%%%%%%%%%%%%%%%%%%%%%%%%%%%%%%%%%%%%%%%%%%%%%%%%%%%%%%%%%%%%%%%%%%%
\section{Introduction}

This paper is aimed at people who wish to define languages and logics
and prove properties about their syntax and semantics fully formally
within an interactive theorem prover. What we present relies upon the
use of dependent types, so we restrict attention to provers providing
those; Agda~\citep{agda}, Lean~\citep{lean} and Rocq~\citep{rocq} are
currently popular examples of such provers. We only need quite basic
features of intensional dependent type
theory~\citep{Martin-LoefP:intttp,NordstromB:progmlt} and here make
use of a subset of the features provided by Agda (see
section~\ref{sec:agda}). The Agda library \agdalink{WSLN} described
here has also been ported to Lean: see
[\href{https://github.com/elefthei/lean-wsln}{github.com/elefthei/lean-wsln}];
a Rocq version should be straightforward.

We assume that the language to be formalized with the theorem prover
involves constructs that bind named resources (logical variables,
storage locations, communication channels, and so on) to lexical
scopes.  Which formalization technique should one choose to deal with
that within the theorem prover? The \textsc{PoplMark}
challenge~\citep{PierceBP:mecmmp} proposed three criteria for
evaluating techniques: \emph{infrastructure overhead} (formalization
overheads, such as additional operations and their associated proof
obligations, should not be prohibitive for large developments),
\emph{transparency} (formal definitions and theorems should not depart
radically from the usual informal conventions familiar to a technical
audience) and \emph{cost of entry} (the formalization technique should
be usable by someone who is not an expert in theorem prover
technology).

Here we will be concerned with the \emph{locally nameless}
representation of syntax involving binding, which dates back to its
mention in the conclusion of~\citep{deBruijnNG:lamcnn}. It uses
explicit names for unbound (free) resources and de~Bruijn indices for
bound ones. For example, the $\alpha$-equivalence class of the untyped
$\lambda$-calculus term $\lambda x.\,\lambda y.\, x(y\,z)$ is
represented by the locally nameless term
$\lambda\,\lambda\,1(0\,z)$. \cite{PierceB:engfm} and
\cite{ChargueraudA:locnr} discuss the history of this form of
representation (they attribute its earliest use in fully formalized
metatheory to \cite{GordonAD:mecncs}) and its relation to other
forms. They make a strong case that the locally nameless technique
satisfies all three \textsc{PoplMark} criteria if it is combined with
cofinite quantification when introducing fresh names in inductive
definitions.  Its \emph{cost of entry} (in the \textsc{PoplMark}
sense) seems lower than the more popular completely nameless
representation, where free names in a term are represented by
``dangling'' de~Bruijn indices, that is, ones that are greater than
the highest level of binding within the term. Using de~Bruijn indices
in this way imposes a spurious order on free names: indices are
totally ordered and have to be carefully shifted up or down as the
context-of-use of the term changes during computation. Whereas all one
should care about a free name is whether it is or is not equal to some
other names; and its identity is invariant under changing its context
of use.

The inductively defined collection of locally nameless terms also
permits dangling indices and one has to cut down to an inductively
defined subset of \emph{locally closed}
terms~\cite[Fig.~1]{PierceB:engfm} to get the things that correctly
represent $\alpha$-equivalence classes of fully nameful terms. This
creates a trap for the unwary. For example, the locally nameless
version of capture-avoiding substitution of terms for free
variables~\cite[Fig.~2]{PierceB:engfm} has a beautifully simple,
structurally recursive definition; but one has to remember that it is
only correct if the term being substituted is locally closed. In
general the locally closed property has to be established many times
in a locally nameless development. For example, when giving an
inductive definition of a type system via a collection of rules, one
has to prove that each rule preserves the locally closed property from
hypotheses to conclusion.

This motivates using the theorem prover's dependent types to enforce
local closure as part of the very definition of locally nameless
terms. In section~\ref{sec:welsln} we give an indexed inductive
definition of types of \emph{$n$-terms}, for each natural number
$n:\Nat$, that contain locally nameless terms which are \emph{well
  scoped} in the sense of only featuring dangling (or free) de~Bruijn
indices drawn from the type $\Fin n$ of natural numbers less than
$n$. Then the $0$-terms play the role of locally closed terms and
correctly represent $\alpha$-equivalence classes of nameful terms (see
section~\ref{sec:ade}). Thus proofs of local closure are replaced by
typing declarations.

This straightforward idea is mentioned by the authors of
\citep[Section~3.3, Footnote~7]{PierceB:engfm}, who do not follow it
up because they want their development to apply also to theorem
provers that are not based on dependent type theory, such as
Isabelle/HOL~\citep{NipkowT:isahpa}. However, the implementation of
the idea is not so straightforward, or very common.\footnote{One
  example occurs in the work at
  [\href{https://github.com/scmu/foundations-harper}{github.com/scmu/foundations-harper}];
  thanks to a referee for pointing to this.} In the conventional
locally nameless representation, de~Bruijn indices are just natural
numbers and the crucial operations of \emph{opening} (replacing an
index with a term) and \emph{closing} (replacing a free name with an
index) have simple, structurally inductive definitions~\citep[Figs~1
and 3]{PierceB:engfm}. In the well-scoped version we have to work a
bit harder since indices are elements of the types $\Fin n$ of finite
ordinals. Opening now involves a function $\Fin(1+n)\fun\Fin n$ that
both removes $i:\Fin(1+n)$ and shifts indices greater than it down by
one; and closing involves the inverse of that. Fortunately we are able
to parameterize our development by an arbitrary binding signature in
the sense of~\cite{AczelP:gencrt} and \cite{PlotkinGD:illtr}. So the
work of implementing the well-scoped version only has to be done once,
by the library, provided the language being formalized can be
specified by such a signature (and it seems that very many can
be). Indeed, in our experience (see section~\ref{sec:exa}) one mostly
needs the opening and closing operations for a few small values of
$n:\Nat$. In case $n=0$, closing an atomic name $x$ with the index
$0:\Fin\,1$ sends a $0$-term $t$ to a $1$-term called its
\emph{abstraction} by $x$ and written $\abs{x}t$; and opening by
replacing index $0$ with a $0$-term $u$ sends a $1$-term $t$ to a
$0$-term called its \emph{concretion} at $u$ and written
$t\conc{u}$. Abstraction, concretion, \emph{substitution}
$(x \ssubst u)\subst t$ (replace occurrences of $x$ in $t$ by $u$) and
\emph{freshness} $x \freshfor t$ (meaning $x$ does not occur in $t$)
have the expected relationships, see Lemma~\ref{lem:conc-abs}, and
quite often that is all one needs to know. Our choice of
abstraction/concretion/freshness terminology reflects the close
relationship \citep{PittsAM:locns} between the locally nameless
representation and nominal techniques
\citep{GabbayMJ:fountl,PittsAM:nomsns}.

The \emph{Well-Scoped Locally Nameless} library \agdalink{WSLN} for
Agda is available at~\citep{welsln-agda} and we give links like
\agdalinkalt{README}{this} to relevant parts of it (converted to html)
throughout this paper. The next section describes the limited set of
features of Agda that it uses. In section~\ref{sec:ato} we explain the
library's use of numbers and finite ordinals for names (atoms) and
scoped de~Bruijn indices respectively.  Section~\ref{sec:welsln}
defines what are the well-scoped locally nameless terms over a given
binding signature and describes the operations on them provided by the
library: scope weakening, freshness, substitution, concretion,
abstraction and size. Section~\ref{sec:ade} establishes a bijection
between such terms and $\alpha$-equivalence classes of ``nameful''
terms, under which substitution corresponds to the usual notion of
capture-avoiding substitution for nameful terms
(Theorems~\ref{thm:bij} and \ref{thm:correct}). Finally, in
section~\ref{sec:exa} we give some examples of the library in use,
allowing us to make various points about the \emph{pros} and
\emph{cons} of this version of the well-scoped locally nameless
representation of syntax.

\subsubsection*{Related work}

There are many techniques for dealing with binding structures when
mechanizing the metatheory of programming languages and logics within
interactive theorem provers. We have already mentioned the
\textsc{PoplMark} challenge of \cite{PierceBP:mecmmp} for evaluating
such techniques; this was revisited by \cite{AbelA:poprmp}.  We will
not survey this wide variety of techniques here, but rather just
concentrate on the use of dependent types to tame dangling de~Bruijn
indices in the locally nameless approach
\citep{GordonAD:mecncs,PierceB:engfm,ChargueraudA:locnr}. However,
that use of dependent types applies just as well to the fully nameless
representation of syntax with binding as it does to the locally
nameless approach. In~\citep{AbelA:decctt}, whose accompanying Agda
code has become an important inspiration for fully formalized proofs
of decidability of conversion in various dependent type theories, the
authors say on page 23:4
\begin{quote}
  ``\emph{the formalization
  does not enforce well-scopedness of expressions, instead we rely on
  the typing judgments to implicitly guarantee well-scopedness. In
  practice this has allowed for some mistakes when formalizing the
  typing rules, which we had to go back and correct, so intrinsically
  well-scoped syntax might have been a better choice}''
\end{quote}
Later developments based on their code convert to a well-scoped
version~\citep{AbelA:gramdt}. So a well-scoped version seems desirable
whether one prefers locally nameless or fully nameless.

The main contribution of this paper is to provide a generic,
signature-parameterised version of the well-scoped locally nameless
style and to do so in a very vanilla version of dependent type theory
(see Section~\ref{sec:agda}). For the non-well-scoped version of
locally nameless, \cite{NiY:genrln} approach genericity in a different
way, building upon the equational axiomatization of
\cite{PittsAM:locns}.

%%%%%%%%%%%%%%%%%%%%%%%%%%%%%%%%%%%%%%%%%%%%%%%%%%%%%%%%%%%%%%%%%%%%%%
\section{Agda}
\label{sec:agda}

We assume some familiarity with intensional dependent type theory as
provided by the Agda interactive theorem prover~\citep{agda}.  We use
Agda's \verb+--safe+ option to switch off possibly inconsistent
experimental features and also its \verb+--without-K+ option since we
do not need uniqueness of identity proofs at all types. In fact we
only need the following very simple form of intensional dependent
type theory, a proper subset of what safe Agda provides:

We use the first two of Agda's hierarchy of universe types,
$\Set : \SetOne$. They are closed under forming dependent function
types and record types (both satisfying the $\eta$-rule for
conversion) and inductive definitions. Agda permits very expressive
forms of the latter (such as inductive-recursive definitions); but we
only need mutual, indexed inductive datatypes and in particular
identity types (written $a \eq b$, where $a$ and $b$ have the same
(implicit) type). We do rely heavily upon Agda's version of dependent
pattern matching~\citep{CoquandT:patmdt}, taking on trust that, in
principle~\citep{GoguenH:elidpm}, its use could be replaced by
eliminators. Agda's implementation of dependent pattern matching with
the \verb+--without-K+ option rules out uniqueness of identity proofs
for arbitrary types~\citep{CockxJ:patmwk}.  Fortunately we only need
such uniqueness at types with decidable equality, where we can appeal
to a theorem of~\cite{HedbergM:cohtml} (see
\agdalink{Prelude.Prop}). We also have no need to assume that
arbitrary functions are extensional with respect to identity types; so
we avoid any use of Agda's irrelevancy annotations, since they
introduce a small amount of extensionality by the back
door.\footnote{With irrelevancy annotations one can prove that any two
  functions $A \fun \bot$ are equal;
  see~\url{https://github.com/agda/agda-stdlib/blob/master/src/Data/Empty.agda}.}
For that reason we avoid use of Agda's Standard
Library~\citep{agdaLib} and instead collect together
a~\agdalink{Prelude} containing a few standard definitions that we
need.

\agdalink{WSLN} makes use of Agda's instance arguments, enclosed by
$\lbbar\_\rbbar$; these are a kind of implicit argument that gets
solved by a special \emph{instance resolution} algorithm, rather than
by the unification algorithm used for normal implicit arguments,
enclosed by single curly braces $\{\_\}$.

%%%%%%%%%%%%%%%%%%%%%%%%%%%%%%%%%%%%%%%%%%%%%%%%%%%%%%%%%%%%%%%%%%%%%%
\section{Atoms and indices}
\label{sec:ato}

We will use elements of a type $\Atom:\Set$, called \emph{atoms},
to name free resources, that is, ones that are not
bound to any scope. We have two requirements for $\Atom$:
\begin{enumerate}
\item It has decidable equality, which is to say that there is a function
  that takes in two atoms $a\;b : \Atom$ and returns either an element
  of $a \eq b$ or a function ${a\eq b}\fun \Empty$ (where $\Empty$ is
  the empty type, a datatype with no constructors).

\item It is \emph{finitely inexhaustible}. Informally this means that
  for every finite set of atoms there is an atom not equal to any of
  them. To make that precise we have to represent finite subsets of
  $\Atom$. The operations on such subsets that we need most often are
  empty set, singletons and unions, so it is helpful (for
  pattern-matching) to make those the constructors of an inductive
  datatype $\FsetA : \Set$ with constructors
  \begin{align*}
    \Empty &:\FsetA &&\text{(empty set)}\\
    \{\_\} &: \Atom\fun\FsetA &&\text{(singleton set)}\\
    \_{\cup}\_ &: \FsetA \fun\FsetA \fun\FsetA &&\text{(binary union)}
  \end{align*}
  Values of $\FsetA$ give a many-one representation of finite
  sets of atoms and we need to be able to define functions
  \[\new :\FsetA\fun\Atom \quad\text{and}\quad
    \newNotIn : (\mathit{xs}:\FsetA)\fun \new\,\mathit{xs} \notin
                \mathit{xs}
  \]
  where $\_{\notin}\_ : \Atom \fun \FsetA \fun\Set$ is a parameterized
  inductive datatype expressing non-membership, which can be defined
  in a straightforward way.
\end{enumerate}
We meet these two requirements by taking $\Atom$ simply to be the
datatype $\Nat$ of natural numbers (with constructors $0: \Nat$ and
$1{+}\_:\Nat\fun\Nat$), and $\new$ to be the function taking a finite
set of numbers to one more than their
maximum; see~\agdalinkalt{WSLN.Atom}{Atom}.

\begin{figure}
  \begin{spacing}{0.9}
    \FIN
    \smallskip
    \ACTFIN
  \end{spacing}
  \mbox{}
  \caption{Indexed datatype of finite ordinals} 
  \label{fig:fin}
\end{figure}

Atoms name free resources, but bound ones are pointed to by de~Bruijn
\emph{indices}. We use ones that are well-scoped in the sense that the
types $\Fin n$ of indices are indexed by natural numbers $n:\Nat$ that
tell us the maximum size of any index currently in scope. Here $\Fin$
is the usual indexed inductive datatype of finite ordinals; see
Figure~\ref{fig:fin}. The values of $\Fin n$ correspond to numbers
less than $n$ via the function $\toN$ defined in the figure.  See
\agdalinkalt{WSLN.Index}{Index} for more details.

%%%%%%%%%%%%%%%%%%%%%%%%%%%%%%%%%%%%%%%%%%%%%%%%%%%%%%%%%%%%%%%%%%%%%%
\begin{figure}
  \begin{spacing}{0.9}
    \LAMBDA
  \end{spacing}
  \mbox{}
  \caption{Binding signature and concrete syntax for untyped
    $\lambda$-calculus} 
  \label{fig:lam}
\end{figure}

\section{Well-scoped locally nameless terms}
\label{sec:welsln}

The Well-Scoped Locally Nameless library (\agdalink{WSLN}) is
parameterized by a binding signature in the sense of
\citep{AczelP:gencrt,PlotkinGD:illtr,FioreMP:abssvb}, a generalization
of the usual notion of algebraic signature:

\begin{definition}[\agdalinkalt{WSLN.Sig.Sig}{Sig}]
  \label{def:sig}
  A \emph{binding signature} $\Sigma$ is given by a type
  $\Op\Sigma : \Set$ whose elements are called \emph{operators}, and a
  function $\arity\Sigma : \Op\Sigma \fun \List\Nat$ assigning a list
  of numbers to each operator.  Given $\op: \Op\Sigma$, the length of
  the list $\arity\Sigma\,\op$ is called the \emph{arity} of the
  operator $\op$ and specifies how many arguments it takes; and each
  number in the list $\arity\Sigma\,\op$ specifies how many names are
  bound in that argument place. 
\end{definition}

\begin{example}[\agdalink{Lambda}]
  \label{exa:lambda}
  An Agda declaration of a binding signature for the untyped
  $\lambda$-calculus is shown in Figure~\ref{fig:lam}. It has two
  operators $\AgdaInductiveConstructor{′ap′}$ (for building
  applications) and $\AgdaInductiveConstructor{′lm′}$ (for building
  $\lambda$-abstractions); $\arity\Sigma$ maps
  $\AgdaInductiveConstructor{′ap′}$ to the list $[0,0]$ (since
  application takes two arguments, neither of which involves binding)
  and maps $\AgdaInductiveConstructor{′lm′}$ to the list $[1]$ (since
  $\lambda$-abstraction takes one argument and is a unary binder).
  Section~\ref{sec:exa} contains some other examples of binding
signatures. 
\end{example}

\begin{quote}
  \textbf{The definitions in the rest of this section are implicitly
    parameterized by a given binding signature $\Sigma$.}
\end{quote}

\begin{definition}[\agdalinkalt{WSLN.Sig.Trm}{Trm}]\label{def:ter}
  The mutually inductive datatypes $\Trm{n}:\Set$ of \emph{$n$-terms}
  (over $\Sigma$) and $\Arg{n}:\List\Nat\fun\Set$ of
  \emph{$n$-argument lists} are indexed by $n:\Nat$ and have
  constructors
  \begin{align*}
    {\idx} &: \Fin n \fun \Trm{n}\\
    {\atm} &: \Atom \fun \Trm{n}\\
    {\oprn} &: (\op : \Op\Sigma) \fun \Arg{n} (\arity\Sigma\,\op) \fun
                \Trm{n}\\
    \NIL &: \Arg{n}\,\NIL\\
    \_{\CONS}\_ &: \{m : \Nat\}\{\mathit{ms} : \List\Nat\} \fun \Trm{m +
                  n} \fun \Arg{n}\,\mathit{ms} \fun \Arg{n}\,(m \CONS
                  \mathit{ms})
  \end{align*}
  ($\NIL$ and $\_{\CONS}\_$ also denote the usual constructors for
  finite lists). Thus an $n$-term is either an index $\idx i$ with
  $i:\Fin n$, or an atomic name $\atm x$ with $x:\Atom$, or a compound
  term $\oprn\op\,[t_1,\ldots,t_k]$ where $k$ is the length of the
  list $\arity\Sigma\,\op : \List\Nat$, say
  $\arity\Sigma\,\op = [m_1,\ldots,m_k]$, and for each $i=1,\ldots,k$,
  $t_i$ is an $(m_i + n)$-term.

  The type of \emph{well-scoped locally nameless terms} is defined to
  be $\Trm{0}$ and we just write it as $\TRM$ (we can do this because
  we defined the indexed datatype to be $\Trm{\_}$ and not $\TRM$).
\end{definition}

\begin{example}
  \label{exa:lamda2}
  Continuing Example~\ref{exa:lambda} with its binding signature for
  the untyped $\lambda$-calculus, the $0$-term corresponding to
  $\lambda x.\,\lambda y.\, x(y\,z)$ is
  \[
    \oprn\AgdaInductiveConstructor{′lm′}\,(
    \oprn\AgdaInductiveConstructor{′lm′}\,(
    \oprn\AgdaInductiveConstructor{′ap′}\,
    (\idx(\suc\,\zero) \CONS
    \oprn\AgdaInductiveConstructor{′ap′}\,(
    \idx\zero\CONS\atm z\CONS\NIL)
    \CONS\NIL)
    \CONS\NIL)
    \CONS\NIL)
  \]
  Evidently the normal forms of $\TRM$ are not very readable; but one
  can introduce helpful abbreviations for indices and for the
  operators of the given signature. For example, with the
  abbreviations defined in Figure~\ref{fig:lam}, the above $0$-term is
  the value of the expression \agdadef{Lambda}{ex'}{ex$'$} =
  $𝛌(𝛌(\mathtt{i1}\cdot(\mathtt{i0}\cdot \atm z)))$.

  The values of $\TRM=\Trm{0}$ in this case correspond to ``locally
  closed'' values of a datatype $\Lambda$ for the (not necessarily
  well-scoped) locally nameless representation: \LN Being locally
  closed is defined by an inductively defined function
  $\mathsf{term}:\Lambda\fun\Set$, which uses a recursively defined
  opening operation ($t,x\mapsto t^x$) and the inductively defined set
  of free variables ($\mathsf{FV}(t)$) of $t:\Lambda$; see
  \cite[Figure~1]{PierceB:engfm}. In the well-scoped version this is
  all replaced by the use of Agda's type system when one makes the
  declaration in Figure~\ref{fig:lam}.
\end{example}

\begin{definition}[\textbf{Scope weakening, support and freshness}]
  \label{def:weasf}
  The function $\actFin$ defined in Figure~\ref{fig:fin} injects
  $\Fin m$ as an initial order segment of $\Fin n$ when $m,n : \Nat$
  and $m \leq n$. This induces functions $\mathtt{Trm}{\wkn}$ and
  $\mathtt{Arg}{\wkn}$ by structural recursion:\\
  \TRMACT\\
  (in the last line \AgdaFunction{+≤} is a proof of $m + n \leq m'
  + n$ given a proof of $m \leq m'$).
  We write $\mathtt{Trm}{\wkn}\,t\,n$ as $t\wkn n$ and call it the
  $n$-term obtained by \emph{scope weakening}. These functions satisfy
  unitary and associative properties that make $\Trm{\_}$ into a
  presheaf on $(\Nat, \geq)$; see~\agdalinkalt{WSLN.Sig.Trm}{Trm}.

  By contrast, weakening $n$-terms with respect to atoms rather than
  indices is an invisible operation, because the finite set of atoms
  currently in context is left implicit in the type of such terms;
  this is one of the strengths of the locally nameless representation
  and carries over to the well-scoped version. Adopting terminology
  from the theory of nominal sets~\citep{PittsAM:nomsns}, we write
  $\supp t : \FsetA$ for the finite set of atoms occurring in
  $t : \Trm{n}$ and call it the \emph{support} of the $n$-term $t$. We
  write $x \freshfor t$ for the type of proofs that an atom $x$ does
  not occur in $\supp t$, in which case we say that $x$ is
  \emph{fresh} for $t$; see \agdalinkalt{WSLN.Fresh}{Fresh}.
\end{definition}

\begin{definition}[\textbf{Substitution} and \textbf{renaming}]
  \label{def:sub}
  We take term-for-name substitutions to be arbitrary functions
  $\Atom\fun\TRM$ mapping names to $0$-terms.\footnote{Using functions
    rather than, for example, finite lists of (atom,term)-pairs,
    conveniently makes the unitary and associative properties of
    composition of substitutions hold up to conversion rather than
    just propositional equality; we just have to be careful to avoid
    properties that depend upon extensional equality of such
    functions.} Since $\Atom$ has decidable equality, single
  substitution is a special case: given $x:\Atom$ and $u:\TRM$, the
  function
  \[
    (x\ssubst u): \Atom\fun\TRM
  \]
  maps $x$ to $u$ and every other name $y$ to the corresponding atom
  $\atm y$. Applying a substitution $\sigma : \Atom\fun\TRM$ to an
  $n$-term $t : \Trm{n}$, we get another $n$-term, denoted
  $\sigma\subst t$. This is defined by a straightforward structural
  recursion on $t$ (see
  \agdalinkalt{WSLN.Sig.Substitution}{Substitution}), using the scope
  weakening operation on terms (Definition~\ref{def:weasf}) for the
  case when $t$ is an atom :
  \begin{align}
    \sigma \subst(\idx i)
    = \; &\idx i\\
    \sigma\subst(\atm x)
    = \; &(\sigma\,x)\wkn n\label{eq:18}\\
    \sigma\subst(\oprn\op\,[t_1,\ldots,t_k])
    = \; &\begin{array}[t]{@{}l}
      \oprn\op\,[\sigma\subst t_1,\ldots, \sigma\subst t_k]\\
      \text{where $\arity\Sigma\,\op = [m_1,\ldots,m_k]$, say.}
    \end{array}
  \end{align}
  For some applications one needs name-for-name \emph{renaming} of
  terms, which we can take to be a special case of substitution: given
  $\rho:\Atom\fun\Atom$ and $t: \Trm{n}$, we define
  $\rho\subst t : \Trm{n}$ to be $(\mathbf{a}\comp\rho)\subst t$ where
  $\mathbf{a}\comp\rho$ is the substitution mapping each $x:\Atom$ to
  $\atm(\rho\, x):\TRM$. A single renaming $(x\ssubst y)\subst t$ is
  defined to be substitution of $\atm{y}$ for $x$,
  $(x\ssubst\atm{y})\subst t$.
\end{definition}

\begin{remark}[Locally nameless versus well-scoped locally nameless.]
  \label{rem:locnws}
  The simplicity of the definition of substitution is one of the
  selling points of the locally nameless representation. For example,
  for the locally nameless datatype $\Lambda$ of $\lambda$-terms in
  Example~\ref{exa:lamda2} the definition would be
  \begin{align*}
    \sigma \subst(\idx i)
    = \; &\idx i\\
    \sigma\subst(\atm x)
    = \; &\sigma\,x\\
    \sigma \subst (𝛌\, t)
    = \; & 𝛌(\sigma\subst t)\\
    \sigma\subst(t\cdot t')
    = \; & (\sigma\subst t)\cdot(\sigma\subst t')
  \end{align*}
  where now $\sigma$ ranges over all functions
  $\Atom\fun\Lambda$. Compared with that, the only minor complication
  for the well-scoped version in Definition~\ref{def:sub} is the need
  in \eqref{eq:18} for scope weakening when $t$ is an atom. However,
  the above locally nameless definition only correctly represents
  capture-avoiding substitution if one restricts the range of a
  substitution to be in the locally closed values of $\Lambda$
  (\emph{cf.}~the use of the $\mathsf{term}$ predicate in the bottom
  half of \cite[Figure~2]{PierceB:engfm}). In the well-scoped version,
  by typing substitutions as functions valued in $\TRM=\Trm{0}$, the
  simple definition of substitution is correct in all cases. We
  substantiate this correctness claim in section~\ref{sec:ade}. To do
  that we first need to consider well-scoped versions of the
  operations of concretion (opening) and abstraction (closing) for
  locally nameless terms.
\end{remark}

Concretion is the $i = \zero$ case of the \emph{opening}
operation~\citep[Figure~1]{PierceB:engfm} which, roughly speaking,
replaces occurrences of an index $i$ in an $(1 + n)$-term by a
$0$-term, to obtain an $n$-term. For the well-scoped version of
locally nameless, opening is not simply a textual replacement, because
as we descend into the arguments of an operation the identity of the
index may change, shifting up by $m$ if the argument is an $m$-ary
binder. So we have to work a little harder than
in~\citep{PierceB:engfm} to give the precise definition, both to get a
well-scoped version and to treat the case of an arbitrary binding
signature:
\begin{definition}[\textbf{Opening} and \textbf{concretion}]
  \label{def:conc} Given $n:\Nat$, the \emph{concretion} $t\conc{u}$
  of an $(1 +n)$-term $t : \Trm{1 + n}$ by a $0$-term $u : \TRM$ is
  the $n$-term $(\zero \open u)t$, where for any index
  $i : \Fin(1 + n)$, the \emph{opening} of $i$ by $u$ in $t$ is the
  $n$-term $(i \open u)t$ defined by recursion on the structure of $t$
  as follows (see~\agdalinkalt{WSLN.Sig.Concretion}{Concretion}):
  \begin{align*}
    (i \open u)(\idx j)
    =
      &\begin{cases}
        u \wkn n &\text{if $i \eq j$}\\
        \idx(\rem{i}{j}) &\text{otherwise}
      \end{cases}\\
    (i\open u)(\atm x)
    = \;&\atm x\\
    (i\open u)(\oprn\op\,[t_1,\ldots,t_k])
    = \; &\begin{array}[t]{@{}l}
      \oprn\op\,[(\suchat{m_1} i \open
      u)t_1,\ldots,(\suchat{m_k} i \open
      u)t_k]\\
      \text{where $\arity\Sigma\,\op = [m_1,\ldots,m_k]$, say.}
    \end{array}
  \end{align*}
  Here $\suchat{m}:\Fin(1 + n) \fun \Fin(1 + m + n)$ is the $m$-fold
  composition of $\suc$ and the function
  $\rem{}{}:(i\;j :\Fin(1+n))\fun \lbbar \_: i \not= j\rbbar \fun \Fin
  n$ removes index $i$ from $Fin(1+n)$ while shifting indices greater than
  it down by one; see~\agdadef{WSLN.Index}{remove}{Index.remove}.
\end{definition}
Abstraction is the $i = \zero$ case of the name \emph{closing}
operation~\citep[Figure~3]{PierceB:engfm} which, roughly speaking,
replaces occurrences of an atom named $x$ in an $n$-term by an index
$i :\Fin (1 + n)$ to obtain a $(1 + n)$-term. As for opening, the
precise definition involves index shifting in binding arguments and
some care in order to get the correct typing for a well-scoped
version:
\begin{definition}[\textbf{Closing} and \textbf{abstraction}]
  \label{def:abs}
  Given $x : \Atom$,
  the \emph{abstraction} $\abs{x}t$ of an $n$-term $t : \Trm{n}$ is
  the $(1+n)$-term $(\zero\close x)t$, where for any index
  $i :\Fin(1 + n)$ the \emph{name closing} of $x$ by $i$ in $t$ is the
  $(1 + n)$-term $(i\close x)t$ defined by recursion on the structure
  of $t$ as follows:
  \begin{align*}
    (i \close x)(\idx j) = \;
    &\idx(\ins{i}{j})\\
    (i\close x)(\atm y) =
    &\begin{cases}
        \idx i &\text{if $x \eq y$}\\
        \atm y &\text{otherwise}
      \end{cases}\\
    (i\close x)(\oprn\op\,[t_1,\ldots,t_k])
    = \; &\begin{array}[t]{@{}l}
      \oprn\op\,[(\suchat{m_1} i \close
      x)t_1,\ldots,(\suchat{m_k} i \close x)t_k]\\
      \text{where $\arity\Sigma\,\op = [m_1,\ldots,m_k]$, say.}
    \end{array}
  \end{align*}
  See~\agdalinkalt{WSLN.Sig.Abstraction}{Abstraction}. As in the
  previous definition $\suchat{m}$ is the $m$-fold composition of
  $\suc$, but now $\ins{i}{j}$ is the result of order-preserving
  insertion of the elements $j$ of $\Fin n$ into $\Fin (1+ n)$
  avoiding the given index $i : Fin (1+ n)$;
  see~\agdadef{WSLN.Index}{insert}{Index.insert}.
\end{definition}

The following lemma gives some properties of concretion, abstraction
and substitution; their proofs may be found in
\agdalinkalt{WSLN.Sig.Concretion}{Concretion},
\agdalinkalt{WSLN.Sig.Abstraction}{Abstraction} and
\agdalinkalt{WSLN.Sig.Substitution}{Substitution}. Unlike for the
equational axiomatization of opening, closing and renaming operations
in \citep{PittsAM:locns}, here no claim of completeness is made for
the properties listed in the lemma; they are just facts that are
needed in Sections~\ref{sec:ade} and \ref{sec:exa}.
\begin{lemma}
  \label{lem:conc-abs}
  \begin{gather}
    {y\freshfor t} \fun {\abs{x}t \eq \abs{y}\,(x\ssubst y)\subst
      t} \label{eq:9}\\
    (\abs{x}t)\conc{u} \eq (x \ssubst u)\subst t\label{eq:1}\\
    {x\freshfor t} \fun \abs{x}\,t\conc{x} \eq t\label{eq:2}\\
    \supp t \subseteq \supp(t\conc{u}) \subseteq \supp t \cup
    \supp u \label{eq:3}\\
    x \notin \supp(\abs{x}t) \subseteq \supp t\label{eq:4}\\
    \sigma\subst (t\conc{u}) \eq (\sigma\subst t)\conc{\sigma\subst
      u} \label{eq:7}\\
    (\forall x \fun {x\in\supp t} \fun {\sigma\,x\eq\sigma'\,x}) \fun
    \sigma* t \eq \sigma'* t \label{eq:13}\\
    \sigma\subst(\abs{x}t) \eq \abs{y}(\sigma\comp(x\ssubst y)
    \subst t)
    \quad\text{if}\;\;\forall z \fun z\in\supp t \fun z\not\equiv x
    \fun y \freshfor
    \sigma\,z \label{eq:8}\\
    (i \open u)(t\wkn(1 +n)) \eq t \wkn n
    \quad\text{if}\;\; i:\Fin(1+n),\, t:\Trm{k},\,
    k\leq\toN(i)\label{eq:5}\\
    (i \close x)(t\wkn n) \eq t \wkn (1 + n)
    \quad\text{if}\;\; i:\Fin(1 + n),\, t:\Trm{k},\, k\leq\toN(i),\,
    x\freshfor
    t\label{eq:6}
  \end{gather}
  \emph{(The function $\toN$ used in \eqref{eq:5} and \eqref{eq:6} was
    defined in Figure~\ref{fig:fin}.)} \qed
\end{lemma}

Although well-scoped locally nameless terms over a binding signature
form an $\Nat$-indexed inductive datatype $\Trm{\_}$, the type
$\Trm{0}$ of $0$-terms (corresponding to locally closed terms in the
non-well-scoped version) is not inductive by itself. So if one wishes
to define a function on that type alone, one cannot proceed by
structural induction without giving a definition for all $n$-terms and
not just for $0$-terms. Instead, one can argue by induction over the
size of $0$-terms, for a suitable notion of size. ``Suitable'' means
that size should be preserved by the operations of scope weakening,
abstraction, concretion with a name, and renaming. Here is one such
definition:

\begin{lemma}
  \label{lem:size}
  We define the \emph{size} of an $n$-term $t: \Trm{n}$ to be the
  natural number $\size\, t : \Nat$ given by recursion on the
  structure of $t$ as follows:
  \begin{align*}
    \size(\idx i) = \;
    &0\\
    \size(\atm x) = \;
    &0\\
    \size(\oprn\op\,[t_1,\ldots,t_k]) = \;
    &1 + \max\{\size\,t_1,\ldots,\size\,t_k\}
  \end{align*}
  Then $\size$ satisfies
  \begin{align}
    \size(t \wkn n) &\eq \size\,t\label{eq:14}\\
    \size(t \conc{x}) &\eq \size\,t\label{eq:15}\\
    \size(\abs{x}t) &\eq \size\,t\label{eq:16}\\
    \size(\rho\subst t) &\eq \size\,t
    \quad(\rho:\Atom\fun\Atom)\label{eq:17}
  \end{align}
  all proved by induction on the structure of $t$;
  see~\emph{\agdalinkalt{WSLN.Sig.Size}{Size}}. \qed
\end{lemma}

The proof of \eqref{eq:12} in Theorem~\ref{thm:bij} is an example of
where induction over the $\size$ of $0$-terms is used to avoid having
to generalise the induction hypothesis to one for $n$-terms for all
$n$. The proof of decidability of $\beta\eta$-conversion for
G{\"o}del's System T (discussed in Example~\ref{exa:gst}) provides
another example: conversion involves the concretion operation
(Definition~\ref{def:conc}), which yields results that are not
sub-terms of the term being concreted; this prevents a proof by
structural induction, but because of \eqref{eq:15} one can induct over
the $\size$ of terms instead.

As mentioned in the Introduction, one of the strengths of the locally
nameless representation (whether or not well-scoped) is its
transparency, in the sense that, compared with a fully nameless
representation, formal definitions and theorems using explicit names
do not depart very much from the usual informal conventions familiar
to a technical audience. However, the use of explicit names does have
some drawbacks, in addition to the issue of size discussed above, for
example in connection with syntax-manipulating computation that
involves the creation of fresh names. In a \emph{relational} setting
this is very neatly handled via the use of cofinite quantification
over names (see the discussion in Example~\ref{exa:mltt}). However,
for \emph{functions} that involve name generation the situation is
less pleasant: in the nameless approach the ``next fresh name'' is
canonical via an index increment, whereas in the nameful approach one
may have to prove explicitly that some property is independent of the
choice of fresh names. Reification and reflection in
Example~\ref{exa:gst} are an example of this (see
\agdalinkalt{GST.ReifyReflect}{ReifyReflect}).

%%%%%%%%%%%%%%%%%%%%%%%%%%%%%%%%%%%%%%%%%%%%%%%%%%%%%%%%%%%%%%%%%%%%%%
\section{Adequacy}
\label{sec:ade}

Why are the data structures and functions described in the previous
section an adequate representation of syntax involving binding
operations? \cite[section~3.4]{PierceB:engfm} discusses what this
question means informally for the locally nameless representation; and
\citep{ChargueraudA:locnr} and its accompanying Rocq development are
convincing evidence that it both works as expected and works
well. Here we wish to provide a formal result about adequacy of the
well-scoped version provided by \agdalink{WSLN}.  The motivation is to
check that the ``index yoga'' present in Definitions~\ref{def:conc}
and~\ref{def:abs} is correct, since it is notoriously easy to make
mistakes with this; also, the author is not aware of a similar result
in the literature at the level of generality of an arbitrary binding
signature. By ``correct'' in the previous sentence we mean that, for
any binding signature, not only is there a bijection between the type
$\TRM$ of well-scoped locally nameless $0$-terms
(Definition~\ref{def:ter}) and $\alpha$-equivalence
classes\footnote{We are working in a very weak dependent type theory,
  in particular one lacking any form of quotient types; so we work
  with $\alpha$-equivalence classes implicitly, using the equivalence
  relation of $\alpha$-conversion between representatives to form a
  setoid.} of nameful terms (Theorem~\ref{thm:bij}), but also the
usual definition of capture-avoiding substitution for such equivalence
classes corresponds to the well-scoped locally nameless definition of
substitution in Definition~\ref{def:sub} (Theorem~\ref{thm:correct}).

First we need to define the ``nameful'' terms associated with a
binding signature $\Sigma$.

\begin{definition}[\agdalinkalt{Adequacy.Nameful}{Nameful}]
  \label{def:nomter}
  The mutually inductive datatypes $\NomTrm:\Set$ (\emph{nameful
    terms}), $\NomArg: \List\Nat \fun \Set$ (\emph{nameful argument
    lists}) and $\NomBnd : (n: \Nat) \fun \Set$ (\emph{nameful $n$-ary
    bindings}) have constructors
  \begin{align*}
    {\atm} &: \Atom \fun \NomTrm\\
    {\oprn} &: (\op : \Op\Sigma) \fun \NomArg\,(\arity\Sigma\,\op) \fun
                \NomTrm\\
    \NIL &: \NomArg\,\NIL\\
    \_{\CONS}\_ &: \{m : \Nat\}\{\mathit{ms} : \List\Nat\} \fun
                  \NomBnd\, m
                  \fun \NomArg\,\mathit{ms} \fun \NomArg\,(m \CONS
                  \mathit{ms}) \\
    \bndNil &: \NomTrm \fun \NomBnd\,0\\
    \bndPair{\_}{\_} &: \{m : \Nat\} \fun \Atom \fun \NomBnd\,m \fun
    \NomBnd\,(1 + m)
  \end{align*}
\end{definition}
Thus a nameful term is either an atom $\atm x$, or a compound term
$\oprn\op\,[b_1,\ldots,b_k]$ where $k$ is the length of the list
$\arity\Sigma\,\op : \List\Nat$, say
$\arity\Sigma\,\op = [m_1,\ldots,m_k]$, and for each $i=1,\ldots,k$,
$b_i$ is a nameful $m_i$-ary binding. A nameful $m$-ary binding takes
the form
$\langle x_m\,,\ldots\langle x_2\,,\langle x_1\,,\langle\rangle M
\rangle\rangle\ldots\rangle$ where the nameful term $M:\NomTrm$ is the
\emph{subject} of the binding and the $m$-vector $x_1,x_2,\ldots,x_m$
of (not necessarily
distinct\footnote{\cite[section~2]{PlotkinGD:illtr} uses distinct
  names; but every nameful term in our more lenient sense is
  $\alpha$-equivalent to one with distinct bound names.}) atomic names
form the \emph{binding names} of the binding.

We translate nameful terms into well-scoped locally nameless terms
using abstraction $\abs{x}\_$ (Definition~\ref{def:abs}) to interpret
name binding $\bndPair{x}{\_}$:

\begin{definition}[\agdalinkalt{Adequacy.Translation}{Translation}]
  \label{def:transl}
  Given a nameful term $M : \NomTrm$, the corresponding well-scoped
  locally nameless term $\den{M}: \TRM$ is defined by structural
  recursion, as follows.
  \begin{align*}
    \den{\atm x}
    &= \atm x\\
    \den{\oprn\op\,[b_1,\ldots,b_k]}
    &= \oprn\op\,[\denb{b_1},\ldots,\denb{b_k}] \\
    \denb{\langle x_m\,,\ldots\langle x_2\,,\langle x_1\,,\langle\rangle
    M \rangle\rangle\ldots\rangle}
    &= \abs{x_m}(\ldots(\abs{x_2}(\abs{x_1} \den{M}))\ldots)
  \end{align*}
\end{definition}

We will prove that $\den{\_}$ gives a bijection between nameful terms
modulo $\alpha$-equivalence and well-scoped locally nameless terms. In
order to define $\alpha$-equivalence for nameful terms we need some
preliminary results about support and renaming for this kind of expression.

The \emph{support} $\supp M : \FsetA$ of a nameful term $M : \NomTrm$
is the finite set of all names occurring in it (be they in atom
position, $\atm{x}$, or in binding position, $\bndPair{x}{\_}$),
defined by a straightforward structural recursion. As before, we write
$x \freshfor M$ for the type of proofs of $x\notin \supp M$ and say
that the name $x$ is \emph{fresh} for $M$ if there is such a proof.

Given a function $\rho:\Atom\fun\Atom$ and $M : \NomTrm$, then the
\emph{renamed} term $\rho\subst M : \NomTrm$ is obtained by applying
$\rho$ to all occurrences of names in $M$, be they in atom position,
$\atm{x}$, or in binding position, $\bndPair{x}{\_}$.  In particular,
if $x,y : \Atom$, then $(x \ssubst y)\subst M$ is the renaming with
$\rho$ the function mapping $x$ to $y$ and all other names to
themselves. Renaming by an arbitrary function $\rho$ may not commute
with the translation in Definition~\ref{def:transl}, for example if
$\rho$ identifies two distinct binding names in a term. However, it
will do so if $\rho$ is a permutation~\citep{PittsAM:nomsns}. We work
implicitly with finite permutations, by considering functions
$\rho:\Atom\fun\Atom$ that are injective on a relevant finite set of
atoms (such as the support of a particular expression). In particular
we have:
\begin{lemma}
  \label{lem:frn}
  For all $x,y:\Atom$ and $M: \NomTrm$, if $y \freshfor M$ then
  $\den{(x\ssubst y)\subst M} \eq (x \ssubst y)\subst \den{M}$
  (and similarly $\denb{(x\ssubst y)\subst b} \eq (x \ssubst y)\subst
  \denb{b}$ if $y \freshfor b : \NomBnd$).
\end{lemma}
\begin{proof}
  This is a corollary of the more general result that
  $\den{\rho\subst M} \eq \rho\subst\den{M}$ if $\rho$ is injective on
  $\supp M$. The proof of that is by induction on the structure of
  $M$, making use of property \eqref{eq:8} from
  Lemma~\ref{lem:conc-abs} for the case of a binder
  $\bndPair{x}{b}$. See
  \agdadef{Adequacy.Translation}{freshRn}{Translation.freshRn}.
\end{proof}

\begin{definition}[\agdadef{Adequacy.Nameful}{_~_}{$\_{∼}\_$}]
  \label{def:aeq}
  The relation of \emph{$\alpha$-equivalence}
  between nameful terms is
  given by datatypes
\begin{align*}
  \_{\aeq}\_
  &: (M\;M' : \NomTrm) \fun \Set\\
  \_{\aeqa}\_
  &: \{\mathit{ms} : \List\Nat\}
  (\mathit{Ms}\;\mathit{Ms}' : \NomArg\,\mathit{ms}) \fun \Set\\
  \_{\aeqb}\_
  &: \{m : \Nat\}(b\;b' : \NomBnd\,m) \fun \Set
\end{align*}
with constructors
\begin{align*}
  {\aeq}{\atm}
  :\; &(x : \Atom)\fun \atm{x} \aeq \atm{x}\\[\jot]
  {\aeq}{\oprn}
  :\; &\begin{array}[t]{@{}l}
    \{\op: \Op\Sigma\}
    \{\mathit{bs}\;\mathit{bs}' : \NomArg\,(\arity\Sigma\,\op)\}
    \fun{}\\
    \mathit{bs}\aeqa\mathit{bs}'
    \fun {\oprn \op\,\mathit{bs}\aeq \oprn\op\,\mathit{bs}'}
  \end{array}\\[\jot]
  {\aeq}\NIL
  :\; &\NIL \aeqa \NIL\\[\jot]
  {\aeq}{\CONS}
  :\; &\begin{array}[t]{@{}l}
    \{m : \Nat\}\{\mathit{ms} : \List\Nat\}\{b\;b' : \NomBnd\,m\}
    \{\mathit{bs}\;\mathit{bs}' : \NomArg\,\mathit{ms}\}\fun{}\\
    b\aeqb b' \fun\mathit{bs}\aeqa\mathit{bs}'\fun
    (\mathit{b}\CONS \mathit{bs}) \aeqa (\mathit{b}'\CONS \mathit{bs}')
  \end{array}\\[\jot]
  {\aeq}\bndNil :\;
      &\{M\;M' : \NomTrm\}\fun
      M \aeq M' \fun \bndNil\,M \aeqb
        \bndNil\,M'\\[\jot]
  {\aeq}\bndPair{}{} :\;
      &\begin{array}[t]{@{}l}
        \{m : \Nat\}\{x\;x'\;y : \Atom\}\{b\;b' : \NomBnd\,m\}\fun{}\\
        {y\freshfor b} \fun {y \freshfor b'} \fun
        {(x\ssubst y)\subst b} \aeqb {(x'\ssubst y)\subst b'}
        \fun \bndPair{x}{b} \aeqb \bndPair{x'}{b'}
      \end{array}
\end{align*}
\end{definition}

\begin{theorem}
  \label{thm:bij}
  The translation in Definition~\ref{def:transl} is one half of a
  bijection between $\NomTrm$ modulo $\aeq$ and $\TRM$ modulo $\eq$.
  In other words, it is \emph{sound} for
  $\alpha$-equivalence, in the sense that for all nameful terms
  $M,M':\NomTrm$
  \begin{equation}
    \label{eq:10}
    {M \aeq M'} \fun {\den{M}\eq \den{M'}}
  \end{equation}
 it is \emph{injective} in the sense that the reverse
 implication
 also holds
 \begin{equation}
   \label{eq:11}
   {\den{M}\eq \den{M'}} \fun {M \aeq M'}
 \end{equation}
 and it is \emph{surjective}, in the sense that for all well-scoped
 locally nameless terms $t : \TRM$ there is a nameful term
 $\den{t}^{-1} : \NomTrm$ with
 \begin{equation}
   \label{eq:12}
   \den{\den{t}^{-1}} \eq t
 \end{equation}
 Taking $t=\den{M}$ in \eqref{eq:12} and applying \eqref{eq:11}, we
 also have $\den{\den{M}}^{-1} \aeq M$.
\end{theorem}
\begin{proof}
  Property~\eqref{eq:10} is proved by induction on the definition of
  $\_{\aeq}\_$ (simultaneously proving similar properties of
  $\_{\aeqa}\_$ and $\_{\aeqb}\_$). The only non-trivial step is for
  the last constructor in Definition~\ref{def:aeq}, where one uses
  Lemmas~\ref{lem:conc-abs}\eqref{eq:9} and
  \ref{lem:frn}.
  See~\agdadef{Adequacy.Translation}{sound}{Translation.sound}.

  Property~\eqref{eq:11} is not as easy to prove. Really the only
  proof strategy available is to do an induction, analysing the
  structure of the terms $M$ and $M'$ in an assumed proof of
  $\den{M}\eq\den{M'}$. The definition of $\den{\_}$ is
  syntax-directed, so all proceeds well until the last induction step,
  for bindings of the form $\bndPair{x}{b}$ and $\bndPair{x'}{b'}$. In
  this case, given the way $\alpha$-equivalence is defined, to
  progress the proof we have to choose a fresh name
  $y = \new(\supp b\cup\supp b')$ in order to apply the induction
  hypothesis to $(x\ssubst y)\subst b$ and $(x'\ssubst y)\subst b'$;
  but these are not sub-expressions of the original terms, so one
  cannot use \emph{structural} induction. Instead we use induction on
  the \emph{size} of terms using functions $\NomTrm\fun\Nat$ and
  $\NomBnd\fun\Nat$ satisfying $\size(\rho\subst M) \eq \size M$ and
  $\size(\oprn\op\,[\denb{b_1},\ldots,\denb{b_k}]) > \max\{\size(b_1)
  , \ldots, \size(b_k)\}$;
  see~\agdadef{Adequacy.Nameful}{sizeNomTrm}{Nameful.sizeNomTrm} for
  the definition of $\size$ and
  \agdadef{Adequacy.Translation}{injective}{Translation.injective}
  for the proof of~\eqref{eq:11}.

  Similarly, for the proof of \eqref{eq:12} we use
  Lemma~\ref{lem:size} to proceed by induction on the size of
  well-scoped locally nameless terms. For the case of a compound
  term $\oprn\op\,[\ldots,t,\ldots]$ we have to construct
  $\den{t}^{-1}$ for an $m$-term $t$ (for some $m:\Nat$) and do so by
  induction on $m$; at the induction step for $1 + m$ we concrete $t$
  at a fresh name $x = \new (\supp\, t)$ to get an $m$-term and rely
  upon the fact that the size of $t\conc{x}$ is the same as the size
  of $t$; see
  \agdadef{Adequacy.Translation}{surjective}{Translation.surjective} for
  the details of the proof.
\end{proof}

To define substitution of nameful terms for the free names of a
nameful term we follow \cite{StoughtonA:subr} and use simultaneous
substitutions and the function $\new :\FsetA\fun\Atom$ that returns a
name that is fresh for a given finite set of names (see
section~\ref{sec:ato}) to get a structurally recursive and totally
defined capture-avoiding substitution function.

\begin{definition}[\agdalinkalt{Adequacy.Substitution}{Substitution}]
  \label{def:capas}
  The nameful term $\sigma\subst M$ obtained from $M:\NomTrm$ by
  \emph{capture-avoiding substitution} along $\sigma:\Atom\fun\NomTrm$
  is defined by recursion on the structure of $M$ as follows:
  \begin{align*}
    \sigma\subst(\atm x) = \;
    &\sigma\,x\\
    \sigma\subst(\oprn\op\,[b_1,\ldots,b_k]) = \;
    &\oprn\op\,[\sigma\subst b_1,\ldots, \sigma\subst b_k]\\
    \sigma\subst(\bndNil\,M) = \;
    &\bndNil\,(\sigma\subst M)\\
    \sigma\subst\bndPair{x}{b} = \;
    &\begin{array}[t]{@{}l}
      \bndPair{y}{\sigma\comp({x:= \atm{y}})\subst b}\\
      \text{where $y =
      \new\textstyle\bigcup\{\supp(\sigma\,z) \mid z\in\supp b\}$}
    \end{array}
  \end{align*}
  The last clause uses the $\new :\FsetA\fun\Atom$ function to select
  an atom $y$ that is not in any of the finite sets $\supp(\sigma\,z)$
  as $z$ ranges over the finite set $\supp b$; the binder
  $\bndPair{x}{\_}$ is freshened to $\bndPair{y}{\_}$ and $\sigma$ is
  updated to $\sigma\comp({x:= \atm{y}})$ before being applied to the
  body $b$.\footnote{\cite{StoughtonA:subr} uses a slightly stronger
    freshness condition, by removing $x$ from the finite set of atoms
    that $y$ has to avoid; the version we have given suffices for
    Theorem~\ref{thm:correct}.}
\end{definition}

\begin{theorem}
  \label{thm:correct}
  The translation $\den{\_}:\NomTrm\fun\TRM$ from
  Definition~\ref{def:transl} sends capture-avoiding substitution of
  nameful terms to substitutions of well-scoped locally nameless
  terms. In other words, if $\sigma:\Atom\fun\NomTrm$ and $M:\NomTrm$,
  then
 \[
  \den{\sigma\subst M} \eq \dens{\sigma}\subst\den{M}
  \]
  where we write $\dens{\sigma}$ for the composition of $\sigma$ with
  $\den{\_}$. In particular, for single substitutions we have
  \[
    \den{(x := M)\subst N} \eq (x := \den{M}) \subst \den{N}
  \]
\end{theorem}
\begin{proof}(See
  \agdadef{Adequacy.Substitution}{*correct}{Substitution.*correct}.)
  The proof is by induction on the structure of $M$. The only
  non-trivial case is the one corresponding to the last clause in the
  definition of capture-avoiding substitution
  (Definition~\ref{def:capas}), where one needs the extensionality
  property~\eqref{eq:13} from Lemma~\ref{lem:conc-abs} together with
  the fact that the translation does not increase support,
  $\supp\den{M}\subseteq\supp M$, which also follows by induction on
  the structure of $M$
  (see~\agdalinkalt{Adequacy.Translation}{Translation}).
\end{proof}

%%%%%%%%%%%%%%%%%%%%%%%%%%%%%%%%%%%%%%%%%%%%%%%%%%%%%%%%%%%%%%%%%%%%%%
\begin{figure}
  \begin{spacing}{0.9}
    \PiC
  \end{spacing}
  \mbox{}
  \caption{Binding signature and concrete syntax for $\pi$-calculus
    processes}
  \label{fig:pic}
\end{figure}

\section{Examples}
\label{sec:exa}

We conclude with some examples of using \agdalink{WSLN}. They allow
us to make various points about the \emph{pros} and \emph{cons} of the
well-scoped locally nameless representation of syntax.

\begin{example}[\agdalink{PiCalc}]
  \label{exa:pic}
  Most of the literature on
  formalizing syntax involving bound names concentrates on languages
  and type systems based on $\lambda$-calculus, where names stand for
  logical variables that can be substituted by terms of the
  language. Programming languages often feature names for other kinds
  of resources that can be bound to syntactic scopes. So to emphasise
  that \agdalink{WSLN} is a library not just for $\lambda$-calculus,
  our first example involves the $\pi$-calculus~\citep{MilnerR:calmp}
  where names stand for channels on which names are communicated. For
  illustrative purposes we consider a core calculus (with process
  replication, but without summation or name matching). The Agda
  declaration of a suitable binding signature is shown in
  Figure~\ref{fig:pic}, together with some concrete syntax to ease
  writing terms over the signature, defined using Agda's
  \AgdaKeyword{pattern} declaration. We give the definitions in full
  to illustrate the common pattern for introducing a signature.

  For this signature, the values of type \AgdaDatatype{Trm} are terms
  that make no distinction between process expressions and channel
  expressions and hence can be ill-formed from the point of view of
  the $\pi$-calculus. For example %
  \Eg %
  is a value of type \AgdaDatatype{Trm} that contains process
  expressions in the first two arguments of %
  \AgdaInductiveConstructor{𝐨}\AgdaSpace{}%
  \AgdaInductiveConstructor{′out′}\AgdaSpace{}%
  where channel names are expected. So an inductively defined
  judgement \AgdaFunction{⊢}$\,P\,$\AgdaDatatype{proc} that $P$ is a
  well-formed process expression is needed;
  see~\agdalink{PiCalc}. One can argue that it would be better
  to let Agda's type system take care of such well-formedness
  conditions as part of the declaration of the signature. To do so
  would require a notion of signature with \emph{sorts}, such as the
  \emph{abstract binding trees} used by \cite{HarperR:prafpl}, or the notion of
  \emph{nominal signature} from \citep{PittsAM:alpsri}, which allows
  for different sorts of name as well as data. We have chosen to use the
  simple, unsorted notion of binding signature from \citep{PlotkinGD:illtr} for
  \agdalink{WSLN}, since in many instances one needs bespoke forms of
  well-formedness (such as some form of type system) that are beyond
  the scope of simple sorting systems.
\end{example}

\begin{example}
  \label{exa:mltt}
  \agdalink{MLTT} develops a small amount of metatheory for a version
  of intensional Martin-L\"of type theory with countably many
  universes closed under dependent function types, identity types and
  possessing a type of natural numbers. A suitable binding signature
  for this is given in \agdalink{MLTT.Syntax}, along with some
  patterns for concrete syntax. The type theory regards types as terms
  of universe type and uses just two forms of judgement:
  \begin{align*}
    \Gamma &\ent a : A : l
    &&\text{$a$ is a term of type $A$ in universe $U_l$}\\
    \Gamma &\ent a = a' : A : l
    &&\text{$a$ and $a'$ are convertible terms of type $A$ in universe $U_l$}
  \end{align*}
  Here $\Gamma$ ranges over a suitable datatype of typing contexts
  (\agdadef{MLTT.Syntax}{Cx}{Cx}), $a,a',A$ are $0$-terms and $l$ is a
  natural number, the universe level; see
  \agdalinkalt{MLTT.Judgement}{Judgement}. The rules for inductively
  generating judgements are the constructors of a datatype
  $\_\AgdaDatatype{⊢}\_$ taking as arguments a typing context and a
  judgement, defined in \agdalinkalt{MLTT.Cofinite}{Cofinite}.

  In existing research on language metatheory, even if a formal
  development uses a nameless form of representation, informal natural
  language descriptions of it in accompanying papers often resort to
  the use of explicit names for the sake of readability. As mentioned
  in the Introduction, one of the strengths of the locally nameless
  representation is its \emph{transparency}, in the sense that formal
  definitions and theorems using it do not depart radically from the
  usual informal conventions familiar to a technical audience (to
  quote the \textsc{PoplMark} challenge~\citep{PierceBP:mecmmp}). The
  same is true for the well-scoped version, with the added bonus that
  it enables Agda's type system to automatically enforce local
  closedness conditions and hence correctness of definitions that
  implicitly depend upon that condition, such as substitution.  Using
  \agdalink{WSLN} the gap between formal definitions and informal
  explanations can be made almost\footnote{Agda's concrete syntax is
    very flexible, but cannot always accommodate informal notations
    that are inherently ambiguous.}  zero with careful choice of
  concrete syntax for the binding signature (as in the bottom half of
  Figure~\ref{fig:pic}, for example). As an example of this, we
  consider the elimination typing rule for natural numbers in
  Martin-L\"of type theory.

  \begin{figure}
    \begin{spacing}{0.9}
      \NELIM
    \end{spacing}
    \mbox{}
    \caption{Typing rule for natural number elimination in ``exists-fresh''
      form}
    \label{fig:nrecef}
  \end{figure}

  \begin{figure}
    \begin{spacing}{0.9}
      \NREC
    \end{spacing}
    \mbox{}
    \caption{Cofinite typing rule for natural number elimination}
    \label{fig:nrec}
  \end{figure}
  
  The signature operation for natural number elimination is called
  \AgdaInductiveConstructor{′natrec′} in \agdalink{MLTT.Syntax} and
  has arity\\
  \ARITY \\
  In other words, it takes four arguments, the first a $1$-term, the
  second and fourth $0$-terms and the third a $2$-term. The pattern
  declaration\\
  \PATTERN\\
  introduces some concrete syntax for terms with this outermost form.
  A possible typing rule for \AgdaInductiveConstructor{′natrec′} is
  shown in Figure~\ref{fig:nrecef}. It is close in form to the
  textbook presentations of \cite[page~64]{NordstromB:progmlt} or
  \cite[A.2.9]{HoTT}. Its last two hypotheses, $x\freshfor(C, c_+)$
  and $y\freshfor c_+$, are needed to ensure that when the rule is
  instantiated, any terms that get substituted for $C$ and $c_+$ do
  not have $x$ or $y$ in their support.  The downside of using
  explicit names for free resources is that they often require such
  explicit assumptions about their freshness (and such requirements
  are often not made explicit in informal accounts).  The rule of
  thumb is that in a hypothesis, whenever an $n$-term appears
  concreted with variables to get a $0$-term, those variables should
  be fresh for the $n$-term and for each other. A good alternative,
  strongly advocated by \cite{PierceB:engfm} and
  \cite{ChargueraudA:locnr}, is to use rules featuring cofinite
  quantification over names, automatically ensuring whatever freshness
  is needed. Such a rule is shown in Figure~\ref{fig:nrec} and is part
  of the definition of the type system in
  \agdalinkalt{MLTT.Cofinite}{Cofinite}. (It uses the notation
  $x\freshfor y \freshfor S$ for the type of proofs that $x$ and $y$
  are distinct names not in the finite set $S$; cf.~\cite[section
  7.1]{ChargueraudA:locnr}.)  Although rules using cofinite
  quantification appear infinitary in nature, they can be proved
  equivalent to the finitary ``exists-fresh'' formulation; see
  \agdalinkalt{MLTT.ExistsFresh}{ExistsFresh}. This is an aspect of
  the ``some/any'' nature of fresh name quantification that is a key
  feature of nominal logic~\cite[section~3.2]{PittsAM:nomsns}.

  The last hypothesis in Figure~\ref{fig:nrec} is marked
  \AgdaComment{helper} because it turns out that, modulo the other
  typing rules, it can be deleted without changing the collection of
  provable judgements. As remarked by \cite[page~23:9]{AbelA:decctt},
  including it and similar helper hypotheses in the other rules in
  \agdalinkalt{MLTT.Cofinite}{Cofinite}, aids proofs by structural
  induction (rather than by induction on a suitable notion of size of
  proofs of judgements).
\end{example}

\begin{example}[\agdalink{GST}]
  \label{exa:gst}
  The final example is a more substantial development than the
  previous two: a proof of decidability of $\beta\eta$-conversion for
  G\"odel's System T~\citep{GoedelK:ubeebn,TaitWW:intiff} using the
  method of normalization-by-evaluation (NBE) within a very weak
  intensional dependent type theory, namely the same one used for
  \agdalink{WSLN} (section~\ref{sec:agda}). Although NBE goes back to
  \cite{BergerU:inveft}, it has been studied and applied by many
  different authors in several different contexts. The development
  here is inspired by the work of \cite{FioreMP:semanbe-jv} on
  categorical analysis of simply typed NBE\footnote{The author
    gratefully acknowledges extended discussions on the topic with
    Fiore.} and by the proof using a semantics in nominal sets given
  in \cite[section~6]{PittsAM:alpsri}. To get for System T (simple
  type theory over a ground type of natural numbers) both a
  semantically-based $\beta\eta$-normalization algorithm and a proof
  of its correctness within intensional type theory we use
  (type-valued) setoids, specifically setoid-valued presheaves on a
  category of renaming functions between typing contexts.

  A binding signature for System~T is given in \agdalink{GST.Syntax}.
  It is not a sub-signature of the one in Example~\ref{exa:mltt},
  because we make all the arguments of the
  \AgdaInductiveConstructor{′natrec′} operator be $0$-terms, typed
  using function types, in order to simplify its semantics. The
  definitions of typing and $\beta\eta$-conversion occur in
  \agdalinkalt{GST.TypeSystem}{TypeSystem}. Types and terms are
  interpreted in a setoid-enriched category of presheaves; see
  \agdalinkalt{GST.Presheaf}{Presheaf},
  \agdalinkalt{GST.TypeSemantics}{TypeSemantics} and
  \agdalinkalt{GST.TermSemantics}{TermSemantics}. The type of natural
  numbers is modelled by the presheaf of normal forms
  (\agdalinkalt{GST.NormalForm}{NormalForm}) of terms of natural
  number type, equipped with the discrete setoid equality given by
  identity types $\_{\eq}\_$. (We take a crude approach and allow
  neutral forms of higher type to be normal forms; consequently a term
  of higher type may have several different normal forms up to
  $\_{\eq}\_$; it would not be difficult to revert to the more usual
  notion of $\eta$-long-$\beta$-normal form.) The semantics of
  \AgdaInductiveConstructor{′natrec′} requires the use of the
  characteristic features of NBE, functions that \emph{reify}
  semantics as normal forms and \emph{reflect} neutral forms to
  semantics; see \agdalinkalt{GST.ReifyReflect}{ReifyReflect}. The
  semantics are proved sound for $\beta\eta$-conversion in
  \agdalinkalt{GST.Sound}{Sound}.  The categorical construction of
  Artin-Wraith glueing (see \cite{FioreMP:semanbe-jv}, for example) is
  left implicit; instead we give a direct definition of a suitable
  logical relation between syntax and semantics
  (\agdalinkalt{GST.LogicalRelation}{LogicalRelation}). Finally, the
  fundamental property of the logical relation and its relationship to
  reification and reflection are used to prove normalization
  (\agdalinkalt{GST.Normalization}{Normalization}) and the desired
  decidability result
  (\agdalinkalt{GST.DecidableConv}{DecidableConv}).

  Freshness of names is a tricky aspect of NBE; see the discussion
  in~\cite[section~2.4]{AbelA:norbed}. We take a nameful approach to
  this issue, requiring the production of fresh names at various
  points; but compared with~\cite[section~6]{PittsAM:alpsri}, we do
  not need to manufacture finitely supported semantic
  objects. Instead, the presheaf semantics keeps track of finite
  support in its Kripke-style ``possible
  worlds''~\citep{AltenkirchT:catrrf}. The emphasis \agdalink{WSLN}
  puts on renaming, rather than the name-permutations used in nominal
  techniques~\citep{PittsAM:nomsns}, is also helpful for the
  semantics. As~\citep[section~5.3]{PopescuA:renrbr} (pre-figured by
  \citep[section~4]{StoughtonA:subr}) points out, compared with the
  permutation-based approach of~\citep[section~6]{PittsAM:alpsri},
  renaming enables a much simpler proof that the semantics of
  abstractions is independent of the choice of fresh name at which
  they are concreted
  (see~\agdalinkalt{GST.TermSemantics}{TermSemantics}).
\end{example}

\begin{remark}[\textbf{Binding structures}] In Example~\ref{exa:pic}
  we mentioned the fact that the kind of binding signatures that
  parameterize \agdalink{WSLN} are unsorted. A more serious limitation
  is that they only model the simple (but useful!)  form of binding in
  which a vector of names of given length is bound in a single
  syntactic scope (cf.~$\NomBnd$ in
  Definition~\ref{def:nomter}). Languages in the wild can feature more
  elaborate forms of binding, with variable numbers of names being
  bound and with several different scopes
  simultaneously. \cite[section~7]{ChargueraudA:locnr} discusses
  locally nameless representation of several such advanced forms of
  binding and it would be useful to develop well-scoped versions of
  them. A good target would be the Ott
  tool~\citep{SewellP:ottets}. This features a meta\-language that
  supports sophisticated binding structures and the tool can compile
  specifications in it into Rocq code using locally nameless
  representation. Considering the translation
  $\den{\_}:\NomTrm\fun \TRM$ from Definition~\ref{def:transl}, one
  can say that Ott replaces $\NomTrm$ with a much more complicated
  type of nameful terms. The question is then, with what should one
  replace the type $\TRM$ of well-scoped locally nameless terms in
  order to translate Ott nameful terms?
\end{remark}

%%%%%%%%%%%%%%%%%%%%%%%%%%%%%%%%%%%%%%%%%%%%%%%%%%%%%%%%%%%%%%%%%%%%%%
% \bibliography{welsln}

\end{document}